%% file: paper.tex
\documentclass[12pt]{article}

\usepackage[T1]{fontenc}
\usepackage{a4wide}
\usepackage{graphics}
\usepackage{color}
\usepackage{amssymb}
\usepackage{amsmath}
\usepackage{amsthm}
		\newtheorem{theorem}{Theorem}
                \newtheorem{lemma}[theorem]{Lemma}
                \newtheorem{corollary}[theorem]{Corollary}

                \newtheorem{definition}[theorem]{Definition}
                
                \newtheorem{observation}[theorem]{Observation}
                \theoremstyle{definition}
                
                \theoremstyle{remark}
                
                \newtheorem*{note*}{Note}
                
                \newtheorem*{remark*}{Remark}
                \theoremstyle{claimstyle}
                
                \newtheorem*{claim*}{Claim}
\usepackage[ruled,linesnumbered]{algorithm2e}
\usepackage[hidelinks]{hyperref}

\newcommand{\True}{\textsf{true}}
\newcommand{\False}{\textsf{false}}
\begin{document}
\title{Parameterized temporal exploration problems\thanks{A preliminary version of this paper appeared in the proceedings of the 1st Symposium on Algorithmic Foundations of Dynamic Networks (SAND 2022), volume 221 of LIPICs, article 15, 2022. DOI 10.4230/LIPIcs.SAND.2022.15}}

\author{Thomas Erlebach\thanks{Department of Computer Science, Durham University. Research supported by EPSRC grants EP/S033483/2 and EP/T01461X/1.}
\and
Jakob T. Spooner\thanks{School of Computing and Mathematical Sciences, University of Leicester.}}

\date{}

\maketitle

\begin{abstract}
In this paper we study the fixed-parameter tractability of the problem of deciding whether a given 
temporal graph $\mathcal{G}$ admits a temporal walk
that visits all vertices (temporal exploration) or, in some problem variants,
a certain subset of the vertices. Formally, a temporal graph is a
sequence $\mathcal{G} = \langle G_1,\ldots, G_L\rangle$ of graphs with $V(G_t) = V(G)$ and $E(G_t) \subseteq E(G)$ for
all $t \in [L]$ and some underlying graph $G$, and 
a temporal walk is a time-respecting sequence of edge-traversals.
We consider both the strict variant, in which edges must be
traversed in strictly increasing timesteps, and the non-strict variant,
in which an arbitrary number of edges can be traversed in each timestep.
For both variants, we give \textsf{FPT} algorithms
for the problem of finding a temporal walk that visits a given set $X$ of
vertices, parameterized by $|X|$, and for the problem of finding a temporal
walk that visits at least $k$ distinct vertices in~$V(G)$, parameterized by~$k$.
We also show $\textsf{W}[2]$-hardness for a set version of the temporal
exploration problem for both variants.
For the non-strict variant, we give an \textsf{FPT} algorithm for
the temporal exploration problem parameterized by the lifetime
of the input graph, and we show that the temporal exploration problem
can be solved in polynomial time if the graph in each timestep has at
most two connected components.
\end{abstract}

\section{Introduction}
The problem of computing a series of consecutive edge-traversals in a static (i.e., classical discrete)
graph $G$, such that each vertex of $G$ is an endpoint of at least one traversed edge, 
is a fundamental problem in algorithmic graph theory, and an early formulation
was provided by Shannon~\cite{Shannon_51}. Such a sequence of edge-traversals might be 
referred to as an \textit{exploration} or \textit{search} of $G$ and, from a computational standpoint, 
it is easy to check whether a given graph $G$ admits such an exploration and easy to compute one if the answer 
is yes -- we simply carry out a depth-first search starting at an arbitrary start vertex
in $V(G)$ and check whether every vertex of $G$ is reached. We consider in this paper
a decidedly more complex variant of the problem, in which we try to find an 
exploration of a \textit{temporal graph}. A temporal graph $\mathcal{G} = \langle G_1,\ldots,G_L\rangle$
is a sequence of static graphs $G_t$ such that $V(G_t) = V(G)$ and $E(G_t) \subseteq E(G)$
for any \textit{timestep} $t \in [L]$ and some fixed \textit{underlying graph}~$G$. 

A concerted effort to tackle algorithmic problems defined for temporal graphs has been made in recent years. With the addition of time to a graph's structure comes more freedom when defining a problem. Hence, many studies have focused on temporal variants of classical graph problems: for example, the travelling salesperson problem \cite{Michail_16}; shortest paths \cite{WCHKLX_14}; vertex cover~\cite{AMSZ_20}; maximum matching \cite{MMNZZ_20}; network flow problems~\cite{ACGKS_19}; and a number of others. For more examples, we point the reader to the works of Molter~\cite{Molter_20} or Michail~\cite{Michail_16}. One seemingly common trait of the problems that many of these studies consider is the following: \textit{Problems that are easy for static graphs often become hard on temporal graphs, and hard problems for static graphs remain hard on temporal graphs}. This certainly holds true for the problem of deciding whether a given temporal graph $\mathcal{G}$ admits a \textit{temporal walk} $W$ -- roughly speaking, a sequence of edges traversed consecutively and during strictly increasing timesteps -- such that every vertex of $\mathcal{G}$ is an endpoint of at least one edge of $W$ (any temporal walk with this property is known as an \textit{exploration schedule}). Indeed, Michail and Spirakis~\cite{MS_16} showed that this problem, \textsc{Temporal Exploration} or \textsc{TEXP} for short, is \textsf{NP}-complete. In this paper, we consider variants of the \textsc{TEXP} problem from a fixed-parameter perspective and under both \textit{strict} and \textit{non-strict} settings. More specifically, we consider problem variants in which we look for \textit{strict} temporal walks, which traverse each consecutive edge at a timestep strictly larger than the previous, as well as variants that ask for \textit{non-strict} temporal walks, which allow an unlimited but finite number of edges to be traversed in  each timestep. 
 
\subsection{Contribution}
\begin{table}
\caption{Overview of results. The parameters are: $L$ = lifetime, $\gamma$ = maximum number of
connected components per step, $k$ = number of vertices to be visited.}
\label{tab:overview}%
\medskip
\small\centering
\begin{tabular}{lccc}
Problem & Parameter & strict & non-strict \\ \hline
\textsc{TEXP} & $L$ & FPT & \textsf{FPT} \\
& & Corollary~\ref{cor:texplFPT} & Theorem~\ref{thm:nstexpl} \\[1ex]
\textsc{TEXP} & $\gamma$ & NPC for $\gamma=1$ & poly for $\gamma=1,2$ \\
& & Observation \ref{obs:gammaOne} & Theorem~\ref{thm:gammaTwo}\\[1ex]
\textsc{$k$-fixed TEXP} & $k$ & \textsf{FPT} & \textsf{FPT} \\
& & Theorem~\ref{thm:kfixedFPT} & Corollary~\ref{cor:NSkfixedFPT} \\[1ex]
\textsc{$k$-arbitrary TEXP} & $k$ & \textsf{FPT} & \textsf{FPT} \\
& & Theorems~\ref{thm:karbFPT}, \ref{thm:karbFPTdet} & Corollary~\ref{cor:NSkarbFPT} \\[1ex]
\textsc{Set-TEXP} & $L$ & $\textsf{W}[2]$-hard & $\textsf{W}[2]$-hard \\
& & Theorem~\ref{thm:sethard} & Theorem~\ref{thm:nssethard}
\end{tabular}%
\end{table}%
An overview of our results is shown in Table~\ref{tab:overview}.
After presenting preliminaries and problem definitions in Section~\ref{sec:prelim},
we show in Section \ref{sec:stricttexp} for the strict setting that two natural parameterized variants of \textsc{TEXP} are in \textsf{FPT}. Firstly, we parameterize by the size $k$ of a fixed subset of the vertex set and ask for an exploration schedule that visits at least these vertices, providing an $O(2^k kLn^2)$-time algorithm. Secondly, we parameterize by only an integer $k$ and ask that a computed solution visits at least $k$ arbitrary vertices -- in this case we specify, for any $\varepsilon > 0$, a randomized algorithm (based on the colour-coding technique first introduced by Alon et al.~\cite{AYZ_95}) with running time $O((2e)^k L n^3 \log \frac{1}{\varepsilon})$. A now-standard derandomization technique~\cite{AYZ_95,NSS_95}
is then utilized in order
to obtain a deterministic $(2e)^k k^{O(\log k)} L n^3 \log n$-time algorithm.
Furthermore, we show that a generalized variant, \textsc{Set TEXP}, in which we are supplied with $m$ subsets of the input temporal graph's vertex set and are asked to decide whether there exists a strict temporal walk that visits at least one vertex belonging to each set, is $\textsf{W}[2]$-hard.

In Section \ref{sec:nstexp}, we consider the non-strict variant known as \textsc{Non-Strict Temporal Exploration}, or \textsc{NS-TEXP}, which was introduced in~\cite{ES_20}. Here, a candidate exploration schedule is permitted to traverse an unlimited but finite number of edges during each timestep, and it is not too hard to see that this change alters the problem's structure quite drastically (more details in Sections \ref{subsec:prelim_nonstrict} and~\ref{sec:nstexp}). We therefore use a different model of temporal graphs to the one considered in Section \ref{sec:stricttexp}, which we properly define later. In this model, an exploration schedule may exist even if the lifetime $L$ is much smaller than the number $n$ of vertices.
Nevertheless, we show that \textsc{NS-TEXP} parameterized by $L$ is \textsf{FPT} by giving an $O(L(L!)^2n)$-time recursive search-tree algorithm.
Furthermore, we show that the \textsf{FPT} algorithms for visiting $k$ fixed vertices or $k$ arbitrary vertices,
where $k$ is taken as the parameter, can be adapted from the strict to the non-strict
case, while saving a factor of $n$ in the running-time.
For the case that the maximum number of components in each step is bounded by~$2$, we show that all four
non-strict problem variants can be solved in polynomial time.
For the non-strict variant of \textsc{Set TEXP}, we show $\textsf{W}[2]$-hardness.

\subsection{Related work}
We refer the interested reader to Casteigts et al. \cite{CFQS_12} for a study of various models of dynamic graphs, and to Michail \cite{Michail_16} for an introduction to temporal graphs and some of their associated combinatorial problems.
Brod\'en et al.~\cite{BHN_04} considered the \textsc{Temporal Travelling Salesperson Problem} for complete temporal graphs with $n$ vertices. The costs of edges are allowed to differ between 1 and 2 in each
timestep. They showed that when an edge's cost changes at most $k$ times during the
input graph's lifetime, the problem is \textsf{NP}-complete,
but provided a $(2-\frac{2}{3k})$-approximation. For the same problem, Michail and
Spirakis~\cite{MS_16} proved \textsf{APX}-hardness and provided a
$(1.7 + \epsilon)$-approximation. Bui-Xuan et
al.~\cite{BFJ_03} proposed multiple objectives for optimisation when computing
temporal walks/paths: e.g., \textit{fastest} (fewest number of timesteps used) and \textit{foremost} (arriving at the destination
at the earliest time possible).

Michail and Spirakis~\cite{MS_16} introduced the \textsc{TEXP} problem, which
asks whether or not a given temporal graph admits a temporal walk that visits
all vertices at least once.
The problem was shown to be
$\textsf{NP}$-complete when no restrictions are placed on the input, 
and they proposed considering the problem under the \textit{always-connected}
assumption as a means of ensuring that exploration is possible (provided the lifetime of
the input graph is sufficiently long). Erlebach et
al.~\cite{EHK_21} considered the problem of computing foremost exploration schedules under the always-connected assumption, proving $O(n^{1-\varepsilon})$-inapproximability (for any $\varepsilon > 0$).
They also showed that subquadratic exploration schedules exist for temporal
graphs whose underlying graph is planar, has bounded treewidth, or is a
$2\times n$~grid. Furthermore, they proved that cycles and cycles with one chord can
be explored in $O(n)$ steps.
Bodlaender and van der Zanden~\cite{BZ_19} examined
the \textsc{TEXP} problem when restricted to always-connected temporal graphs whose underlying graph has pathwidth at most 2, showing the problem to be $\textsf{NP}$-complete in this case. 

Later, Erlebach et al.~\cite{EKLSS/19} showed that temporal graphs can be explored
in $O(n^{1.75})$ steps if the graph in each step admits a spanning-tree of
bounded degree or if one is allowed to traverse two edges per step.
Taghian Alamouti~\cite{A/20} showed that a cycle with $k$ chords can be
explored in $O(k^2\cdot k!\cdot (2e)^k \cdot n)$ timesteps. Adamson et
al.~\cite{AGMZ/22} improved this bound for cycles with $k$ chords to $O(kn)$ timesteps.
They also improved the
bounds on the worst-case exploration time for temporal graphs whose underlying
graph is planar or has bounded treewidth.

Akrida et al.~\cite{AMSR_21} considered a \textsc{TEXP} variant called
\textsc{Return-To-Base TEXP}, in which the underlying graph is a star and a candidate solution must return to the
vertex from which it initially departed (the star's centre). They proved various 
hardness results and provided polynomial-time algorithms for some special cases.
Casteigts et al.~\cite{CHMZ_21} studied the fixed-parameter tractability of the problem of finding 
temporal paths between a source and destination that wait no longer than $\Delta$ consecutive timesteps at any intermediate vertex.
Bumpus and Meeks~\cite{BM_21} considered, again from a fixed-parameter perspective, a temporal
graph exploration variant in which the goal is no longer to visit all of the input graph's vertices at least once, 
but to traverse all edges of its underlying graph exactly once (i.e., computing a temporal Eulerian circuit). They also resolved the complexity of the two cases of the \textsc{Return-To-Base TEXP}
problem that had been left open by~\cite{AMSR_21}.

The problem of \textsc{Non-Strict Temporal Exploration} was introduced and 
studied in~\cite{ES_20}. Here, a computed walk may make an unlimited
number of edge-traversals in each given timestep. Amongst other things,  
\textsf{NP}-completeness of the general problem was shown, as well as $O(n^{1/2-\varepsilon})$ and 
$O(n^{1-\varepsilon})$-inapproximability for the problem of minimizing the arrival time
of a temporal exploration in the cases
where the number of timesteps required to reach any vertex $v$ from any vertex $u$ is bounded by $c=2$ and $c=3$, respectively. Notions of strict/non-strict
paths which respectively allow for a single edge/unlimited number of edge(s) to
be crossed in any timestep have been considered before, notably by Kempe et
al.~\cite{KKK_00} and Zschoche et al.~\cite{ZFMN_20}.

\section{Preliminaries}
\label{sec:prelim}%
For a pair of integers $x,y$ with $x \leq y$ we denote by $[x,y]$ the set $\{z : x \leq z \leq y\}$; if $x = 1$ we write $[y]$ instead. 
We use standard terminology from graph theory~\cite{Diestel_00}, and we
assume any static graph $G = (V,E)$ to be simple and undirected.
A parameterized problem is a language $L\subseteq \Sigma^*\times\mathbb{N}$,
where $\Sigma$ is a finite alphabet. For an instance $(I,k)\in \Sigma^*\times\mathbb{N}$,
$k$~is called the parameter. The problem is in \textsf{FPT} (fixed-parameter tractable)
if there is an algorithm that solves every instance in time $f(k)\times |I|^{O(1)}$
for some computable function~$f$. A proof that a problem is hard for complexity class
$\textsf{W}[r]$ for some integer $r\ge 1$ is seen as evidence that the problem is unlikely
to be contained in \textsf{FPT}.
For more on parameterized complexity, including definitions of the complexity classes
$\textsf{W}[r]$, we refer to \cite{DF_99,Cygan_et_al/15}.

\subsection{Temporal exploration with strict temporal walks}
\label{subsec:prelim_strict}%
The relevant concepts and problem definitions for strict temporal walks are as follows.
We begin with the definition of a temporal graph:

\begin{definition}[Temporal graph]\label{def:tempgraph}
A temporal graph $\mathcal{G}$ with $underlying$ $graph$ $G = (V, E)$, lifetime $L$ and order $n$ is a sequence of simple undirected graphs $\mathcal{G} = \langle G_1,G_2,\ldots,G_L \rangle$ such that $|V| = n$ and $G_t = (V, E_t)$ (where $E_t \subseteq E$) for all $t \in [L]$. 
\end{definition}
For a temporal graph $\mathcal{G} = \langle G_1,\ldots,G_L\rangle$, the subscripts $t \in [L]$ indexing the graphs in the sequence are referred to as \textit{timesteps} (or \textit{steps}) and we call $G_t$ the $t$-th \emph{layer}. A tuple $(e,t)$ with $e \in E(G)$ is an \textit{edge-time pair} (or \textit{time edge\/}) of $\mathcal{G}$ if $e \in E_t$. Note that the size of any temporal graph (i.e., the maximum number of time edges) is  bounded by $O(Ln^2)$.
\begin{definition}[Strict temporal walk]\label{def:tempwalk}
	A strict temporal walk $W$ in $\mathcal{G}$ is a tuple	$W = (t_0, S)$,
 	consisting of a start time $t_0$ and an alternating sequence of vertices and edge-time pairs $S = \langle v_1,(e_1,t_1),v_2,(e_2,t_2),\ldots,v_{l-1},(e_{l-1},t_{l-1}),v_l\rangle$ such that 
 	$e_i = \{v_i,v_{i+1}\}$, $e_i \in G_{t_i}$ for $i\in [l-1]$ and $1 \leq t_0 \leq t_1 < t_2 < \cdots < t_{l-1} \leq L$. 
\end{definition}
We say that a temporal walk $W = (t_0, S)$ \textit{visits} any vertex that is included in~$S$. Further, $W$ \textit{traverses} edge $e_i$ at time $t_i$ for all $i \in [l-1]$ and is said to \textit{depart
 	from} (or start at) $v_1 \in V(\mathcal{G})$ at timestep $t_0$ and \textit{arrive at} (or finish at) 
 	$v_l \in V(\mathcal{G})$ at the end of timestep $t_{l-1}$ (or, equivalently, at the beginning of timestep $t_{l-1}+1$).
	Its \emph{arrival time} is defined to be $t_{l-1}+1$. It is assumed that $W$ is positioned at $v_1$ at 
 	the start of timestep $t_0 \in [t_1]$ and waits at $v_1$ until edge $e_1$ is traversed during timestep $t_1$.
 	The quantity $|W| = t_{l-1} - t_0 + 1$ is called the \emph{duration} of $W$.
Observe that the arrival time of a strict temporal walk
equals its start time plus its duration.
We remark that a walk with arrival time $t$ that finishes at a vertex~$v$
and a walk with start time $t$ (or later) that departs from~$v$ can
be combined into a single walk in the obvious way.

We denote by $sp(u,v,t)$ the duration of a shortest (i.e., having minimum arrival time) temporal walk in $\mathcal{G}$ that starts at $u \in V(\mathcal{G})$ in timestep $t$ and ends at $v \in V(\mathcal{G})$. If $u=v$, $sp(u,v,t)=0$. We note that there is no guarantee that a walk between a pair of vertices $u,v$ exists; in such cases we let $sp(u,v,t) = \infty$. The algorithms that we present in Section~\ref{sec:strict} will repeatedly require us to compute such shortest walks for specific pairs of vertices $u,v \in V(\mathcal{G})$ and a timestep $t \in [L]$ -- the following theorem allows us to do this:
\begin{theorem}[Wu et al.~\cite{WCHKLX_14}]\label{thm:shortpaths}
	Let $\mathcal{G} = \langle G_1,\ldots,G_L\rangle$ be an arbitrary temporal graph. Then, for any $u \in V(\mathcal{G})$ and $t \in [L]$, one can compute in $O(Ln^2)$ time for all $v \in V(\mathcal{G})$ 
	the value $sp(u,v,t)$. For any $v\in V(\mathcal{G})$ for which $sp(u,v,t)$ is finite,
	a temporal walk that starts at $u$ at time $t$, ends at $v$, and has duration $sp(u,v,t)$
	can then be determined in time proportional to the number of time-edges of that walk.
\end{theorem}
The following two definitions will be used to describe the sets of candidate solutions for several of the
problems that we consider in this paper.
\begin{definition}[$(v,t,X)$-tour]\label{def:vtxtour}
A $(v,t,X)$-tour $W$ in a given temporal graph $\mathcal{G}$ is a strict temporal walk that
starts at some vertex $v \in V(\mathcal{G})$ in timestep $t$ and visits (at least) all vertices in 
$X \subseteq V(\mathcal{G})$. We can assume that the walk ends as soon as
all vertices in $X$ have been visited, so we take the arrival time $\alpha(W)$ of a $(v,t,X)$-tour $W$
to be the timestep after the timestep at the end of which $W$
has for the first time visited all vertices in~$X$.
\end{definition}
\begin{definition}[$(v,t,k)$-tour]\label{def:vtktour}
	A $(v,t,k)$-tour $W$ in a given temporal graph $\mathcal{G}$ is a $(v,t,X)$-tour
for some subset $X \subseteq V(\mathcal{G})$ that satisfies $|X| = k$. The arrival time
$\alpha(W)$ of a $(v,t,k)$-tour $W$ is the timestep after the timestep at the end of which $W$
has for the first time visited all vertices in~$X$.
\end{definition}
A $(v,t,X)$-tour $W$ ($(v,t,k)$-tour $W^*$) in a temporal graph $\mathcal{G}$ is said to be 
\textit{foremost} if $\alpha(W) \leq \alpha(W')$
($\alpha(W^*) \leq \alpha({W^*}')$) for any other $(v,t,X)$-tour $W'$ (any other $(v,t,k)$-tour ${W^*}'$).
We now formally define the main problems of interest:
For a given temporal graph $\mathcal{G}$ with start vertex $s\in V(\mathcal{G})$,
an $(s,1,V)$-tour is also called an \emph{exploration schedule}.
The standard temporal exploration problem is defined as follows:

\begin{definition}[\textsc{TEXP}]\label{def:texp}
An instance of \textsc{TEXP} is given as a tuple $(\mathcal{G},s)$, where $\mathcal{G}$ is an
arbitrary temporal graph with underlying graph $G=(V,E)$ and lifetime~$L$;
and $s$ is a start vertex in $V(\mathcal{G})$.
The problem then
asks that we decide if there exists an exploration schedule in $\mathcal{G}$.
\end{definition}

Instead of visiting all vertices, we may be interested in visiting all vertices
in a given set of $k$ vertices, or even an arbitrary set of $k$ vertices. These
problems are captured by the following two definitions.

\begin{definition}[\textsc{$k$-fixed TEXP}]\label{def:kfixedtexp}
	An instance of the \textsc{$k$-fixed TEXP} problem is given as a tuple 
$(\mathcal{G},s,X,k)$ where $\mathcal{G}=\langle G_1,\ldots,G_L\rangle$ is an
arbitrary temporal graph with underlying graph $G$ and lifetime $L$;
 $s$ is a start vertex in $V(\mathcal{G})$; and 
$X \subseteq V(\mathcal{G})$ is a set of target 
vertices such that $|X| = k$.
The problem then
asks that we decide if there exists an $(s,1,X)$-tour $W$ in $\mathcal{G}$.
\end{definition}
\begin{definition}[\textsc{$k$-arbitrary TEXP}]\label{def:karbtexp}
	An instance of the \textsc{$k$-arbitrary TEXP} problem is given as a tuple $(\mathcal{G},s,k)$ 
where $\mathcal{G}=\langle G_1,\ldots,G_L\rangle$ is an arbitrary 
temporal graph with underlying graph $G$ and lifetime $L$; $s$ is a start vertex in 
$V(\mathcal{G})$; and $k \in \mathbb{N}$. 
The problem then asks that we decide whether there exists an $(s,1,k)$-tour $W$ in $\mathcal{G}$.
\end{definition}

Finally, we may be given a family of subsets of the vertex set, and our goal may be
to visit at least one vertex in each subset. This leads to the following problem,
whose definition is analogous to the \textsc{Generalized TSP} problem~\cite{Noon/88}
(also known by various other names including \textsc{Set TEXP}, \textsc{Group TSP},
and \textsc{Multiple-Choice TSP}).

\begin{definition}[\textsc{Set TEXP}]
An instance of \textsc{Set TEXP} is given as a tuple $(\mathcal{G},s,\mathcal{X})$, 
where $\mathcal{G}$ is an arbitrary temporal graph with lifetime $L$, $s \in V(\mathcal{G})$ is a 
start vertex, and $\mathcal{X} = \{X_1,\dots,X_m\}$ is a set of subsets $X_i \subseteq V(\mathcal{G})$. 
The problem then asks whether or not there exists a set $X\subseteq V(\mathcal{G})$ and
an $(s,1,X)$-tour in $\mathcal{G}$ with $X\cap X_i\neq\emptyset$ for all $i\in [m]$.
\end{definition}

For yes-instances of all the problems defined above, a tour with
minimum arrival time (among all tours of the type sought) is called an \emph{optimal solution}.

\subsection{Temporal exploration with non-strict temporal walks}
\label{subsec:prelim_nonstrict}%
When we consider the non-strict version of \textsc{TEXP}, a walk
is allowed to traverse an unlimited number of edges in every timestep. As mentioned in the introduction, this changes the nature of the problem significantly. In particular, it means that a temporal walk positioned at a vertex $v$ in timestep $t$ is able to visit, during timestep $t$, any other vertex contained in the same connected component $C$ as $v$ and move to an arbitrary vertex $u \in C$, beginning timestep $t+1$ positioned at vertex $u$. As such, it is no longer necessary to know the edge structure of the input temporal graph during each timestep, and we can focus only on the connected components of each layer. This leads to the following definition:
\begin{definition}[Non-strict temporal graph, $\mathcal{G}$]\label{nstg}
	A non-strict temporal graph $\mathcal{G} = \langle G_1, \ldots, G_L\rangle$ with vertex set $V := V(\mathcal{G})$ and lifetime $L$ is an indexed sequence of partitions (layers) $G_t = \{C_{t,1}, \ldots, C_{t,\gamma_t}\}$ of $V$ for $t \in [L]$. For all $t \in [L]$, each $v \in V$ satisfies $v \in C_{t,j}$ for a unique $j \in [\gamma_t]$. The integer $\gamma_t$ denotes the number of components in layer $G_t$; clearly we have $\gamma_t \in [n]$.
\end{definition}
For a given non-strict temporal graph with lifetime $L$ and
$\gamma_t$ components per step for $t\in[L]$, we define
$\gamma=\max_{t\in[L]}\gamma_t$ to be the \emph{maximum number of
components per step}.
A non-strict temporal walk is defined as follows:
\begin{definition}[Non-strict temporal walk, $W$]\label{nstw}
	A non-strict temporal walk $W$ starting at vertex $v$ at time $t_1$ in a non-strict temporal graph $\mathcal{G} = \langle G_1, \ldots, G_L\rangle$ is a sequence $W = C_{t_1,j_1},C_{t_2,j_2},\ldots,C_{t_l,j_l}$ of components $C_{t_i,j_{i}}$ ($i \in [l]$) with $1\le t_1 \le t_l\le L$ such that: $t_i+1 = t_{i+1}$ for all $i \in [1,l-1]$; $C_{t_i, j_i} \in G_{t_i}$ and $j_{i} \in [\gamma_{t_i}]$ for all $i \in [l]$; $C_{t_i, j_{i}} \cap C_{t_{i+1},j_{i+1}} \neq \emptyset$ for all $i \in [l-1]$; and $v \in C_{t_1,j_1}$. Its
	arrival time is defined to be~$t_l$.
\end{definition} 
Let $W = C_{t_1,j_{1}},C_{t_2,j_{2}},\ldots,C_{t_l,j_{l}}$ be a non-strict temporal walk in some 
non-strict temporal graph $\mathcal{G}$ starting at some vertex $s \in C_{t_1,j_1}$. We refer to $l-1$ as the \textit{duration} of $W$. The walk $W$ is said to start at vertex $s \in C_{t_1,j_{1}}$ in timestep $t_1$ and finish at component $C_{t_l,j_{l}}$ (or sometimes at some $v \in C_{t_l,j_l}$) in timestep $t_l$. Furthermore, $W$ \textit{visits} the set of vertices $\bigcup_{i \in [l]} C_{t_i,j_{i}}$. Note that $W$ visits exactly one component in each of the $l$ timesteps from $t_1$ to $t_l$. We call $W$ \textit{non-strict exploration schedule starting at $s$} with \textit{arrival time} $l$ if $t_1 = 1$ and $\bigcup_{i \in [l]} C_{t_i,j_{i}} = V(\mathcal{G})$.
A non-strict temporal walk $W_1$ that finishes in component $C_{t,j}$
and a non-strict temporal walk $W_2$ that starts at a vertex $v$ in $C_{t,j}$
at time~$t$ can be combined into a single non-strict temporal walk in the obvious way.
This is why the arrival time of $W_1$ is defined to be $t$ rather than~$t+1$,
as one might have expected in analogy with the case of strict temporal walks.
Furthermore, note that the arrival time of a non-strict temporal walk
equals its start time plus its duration.

A \emph{non-strict $(v,t,X)$-tour} is a non-strict temporal walk that starts at $v$ at time $t$
and visits at least all vertices in~$X$.
A \emph{non-strict $(v,t,k)$-tour} is a non-strict $(v,t,X)$-tour for some $X\subseteq V$
with $|X|=k$.

The problems
\textsc{TEXP},
\textsc{$k$-fixed TEXP},
\textsc{$k$-arbitrary TEXP}, and
\textsc{Set TEXP} that have been defined for strict temporal walks
then translate into the corresponding problems for non-strict temporal walks,
which we call
\textsc{NS-TEXP},
\textsc{$k$-fixed NS-TEXP},
\textsc{$k$-arbitrary NS-TEXP}, and
\textsc{Set NS-TEXP}, respectively.

\section{Strict TEXP parameterizations}\label{sec:stricttexp}
\label{sec:strict}%
In this section, we consider temporal exploration problems
in the strict setting.
First, we observe that we cannot hope for
an \textsf{FPT} algorithm for \textsc{TEXP} for parameter~$\gamma$, the maximum number of connected
components per step, unless $\mathsf{P}=\mathsf{NP}$:
It was shown in \cite[Theorem 3.5]{EHK_21} that \textsc{TEXP} is \textsf{NP}-hard
even if the graph in each timestep is the same connected planar
graph of maximum degree~$3$, which implies the following:

\begin{observation}\label{obs:gammaOne}
\textsc{TEXP} is \textsf{NP}-hard even if $\gamma=1$.
\end{observation}

In the remainder of this section, we first give an \textsf{FPT} algorithm
for \textsc{$k$-fixed TEXP} in Section~\ref{ss:kfixed}.
In Section~\ref{ss:karb}, we first give a randomized
\textsf{FPT} algorithm for \textsc{$k$-arbitrary TEXP}
and then show how to derandomize it.
In Section~\ref{ss:set}, we show that \textsc{Set TEXP}
is $\textsf{W}[2]$-hard for parameter~$L$.

\subsection{An \textup{\textsf{FPT}} algorithm for \texorpdfstring{\textsc{$k$-fixed TEXP}}{k-fixed TEXP}}\label{ss:kfixed}
In this section we provide a deterministic \textsf{FPT} algorithm for \textsc{$k$-fixed TEXP}.
Let $(\mathcal{G},s,X,k)$ be an instance of \textsc{$k$-fixed TEXP}.
For a given order $(v_1,v_2,\ldots,v_k)$ of $k$ vertices,
one can use Theorem~\ref{thm:shortpaths}
to check in polynomial time whether it is possible to visit the vertices
in that order: We find the earliest arrival time for reaching $v_1$
from~$s$, then the earliest arrival time for reaching $v_2$ from $v_1$
if we start at $v_1$ at the arrival time of the first walk,
and so on. In this way we obtain a walk that visits the vertices in
the given order, if one exists, and that walk has earliest arrival time
among all such walks.
Therefore, one approach to obtaining an \textsf{FPT} algorithm
for \textsc{$k$-fixed TEXP} would be to enumerate all $k!$ possible orders
in which to visit the $k$ vertices, and to determine for each order using
Theorem~\ref{thm:shortpaths} whether it is possible to visit the vertices
in that order. In the following, we design an \textsf{FPT} algorithm
for \textsc{$k$-fixed TEXP} whose running-time has a better
dependency on~$k$, namely, $2^kk$ instead of $k!$.

Our algorithm looks for an 
earliest arrival time $(s,1,X)$-tour of $\mathcal{G}$ via a dynamic programming (DP) approach. 
We note that the approach is essentially an adaptation of an algorithm proposed (independently by Bellman~\cite{Bellman_62} and Held \& Karp~\cite{HK_62}) for the classic Travelling Salesperson Problem to the parameterized problem for temporal graphs.
\begin{theorem}\label{thm:kfixedFPT}
	It is possible to decide any instance $I = (\mathcal{G},s,X,k)$ of \textsc{$k$-fixed TEXP}, and return an optimal solution if $I$ is a yes-instance, in time $O(2^kkLn^2)$, where $n = |V(\mathcal{G})|$ and $L$ is $\mathcal{G}$'s lifetime.
\end{theorem}
\begin{proof}		
	First we describe our algorithm before proving its correctness and analysing its running time. We begin
	by specifying a dynamic programming formula for $F(S,v)$, by which we denote the minimum arrival time of any temporal walk in $\mathcal{G}$ that starts at vertex $s \in V(\mathcal{G})$ in timestep $1$, visits all vertices in $S \subseteq X$, and finishes at vertex $v \in S$. One can compute $F(S,v)$ via the following formula:
		\begin{eqnarray}\label{eqn:fsv}  
			F(S,v) = 
     			\begin{cases}
       				1+sp(s,v,1)		 &\; (|S| = 1)\\
       				\displaystyle \min_{u \in S-\{v\}} [ F(S-\{v\},u) + sp(u,v,F(S-\{v\},u))] &\; (|S| > 1)\\
     			\end{cases}
		\end{eqnarray}
		Note that to compute $F(S,v)$ when $|S|>1$, Equation (\ref{eqn:fsv}) states that we need only consider values $F(S',u)$ with $u \in S'$ and $|S'| = |S|-1$, and so we begin by computing all values $F(S',u)$ such that $S' \subseteq X$ satisfies $|S'| = 1$ and $u \in S'$, before computing all values such that $|S'| = 2$ and $u \in S'$ and so on, until we have computed all values $F(X,u)$ where $u \in X$ (i.e., values $F(S',u)$ with $|S'| = k = |X|$). Once all necessary values have been obtained, computing the following value gives the arrival time of an optimal $(s,1,X)$-tour: 
		\begin{eqnarray}\label{eqn:fstar}
			F^* = \min_{v \in X} F(X,v).
		\end{eqnarray} 
		If, whenever we compute a value $F(S,v)$ with $|S| > 1$, we also store alongside $F(S,v)$ a single pointer 
		\[p(S,v) = \arg\min_{u \in S-\{v\}} [ F(S-\{v\},u) + sp(u,v,F(S-\{v\},u))] ,\] 
		then once we have computed $F^*$ we can use a traceback procedure to reconstruct the walk with arrival time $F^*$. More specifically, let $u_1 = \arg\min_{u \in X} F(X,u)$ and $u_i = p(X-\{u_1,\ldots,u_{i-2}\},u_{i-1})$ for all $i \in [2,k]$. To complete the algorithm, we then check if $F^*$ is finite: If so, then there must be a $(s,1,X)$-tour $W$ in $\mathcal{G}$ with $\alpha(W) = F^*$ that visits the vertices $u_{k},\ldots,u_{1}$ in that order. We can reconstruct $W$ by concatenating the $k$ shortest walks obtained by starting at $s$ in timestep $1$ and computing a shortest walk from $s$ to $u_{k}$, then computing a shortest walk from $u_{k}$ to $u_{k-1}$ starting at the timestep at which $u_{k}$ was reached, and so on, until $u_1$ is reached; once constructed, return $W$. If, on the other hand, $F^* = \infty$ (which is possible by the definition of $sp(u,v,t)$) then return no.
		
\paragraph{Correctness}
The correctness of Equation (\ref{eqn:fsv}) can be shown via induction on $|S|$: The base case (i.e., when $|S| = 1$) is correct since the arrival time of the foremost temporal walk that starts at $s$ in timestep $1$ and ends at a specific vertex $v \in X$ is clearly equal to one plus the duration of the foremost temporal walk between $s$ and $v$ starting at timestep $1$. 
		
		For the general case (when $|S| > 1$), assume first that the formula holds for any set $S'$ such that $|S'| = l$ and any vertex $u \in S'$. To see that the formula holds for all sets $S$ with $|S| = l+1$ and vertices $v \in S$, consider any walk $W$ that starts in timestep 1, visits all vertices in some set $S$ with $|S| = l+1$ and ends at $v$. Let $x_1,\ldots,x_{l+1}$ be the order in which the vertices $x_i \in S$ are reached by $W$ for the first time; let $x=x_{l+1}=v$ and $x'=x_{l}$. Note that the subwalk $W'$ of $W$ that begins in timestep 1 and finishes at the end of the timestep in which $W$ arrives at $x'$ for the first time is surely an $(s,1,S-\{v\})$-tour, since $W'$ visits every vertex in $S-\{x\} = S-\{v\}$. Then, by the induction hypothesis we have $\alpha(W') \geq F(S-\{v\},x')$ because $|S-\{v\}| = l$, and since $W$ ends at $v$ we have
		\begin{eqnarray*}
			\alpha(W) & \geq & \alpha(W') + sp(x',v,\alpha(W'))  \\
				      & \geq & F(S-\{v\},x') + sp(x',v,F(S-\{v\},x')).
		\end{eqnarray*}
More generally, we can say that any $(s,1,S)$-tour $W$ that starts at $s$ in timestep 1, visits all vertices in $S$ (where $|S| = l+1$), and finishes at $v \in S$ satisfies the above inequality for some $x' \in S-\{v\}$. Note that for any $u \in S-\{v\}$, $F(S-\{v\},u) + sp(u,v,F(S-\{v\},u))$ corresponds to the arrival time of a valid $(s,1,S)$-tour, obtained by concatenating an earliest arrival time $(s,1,S-\{v\})$-tour that ends at $u$ and a shortest walk between $u$ and $v$ starting at time $F(S-\{v\},u)$. Therefore, to compute $F(S,v)$ it suffices to compute the minimum value of $F(S-\{v\}, u) + sp(u,v,F(S-\{v\},u))$ over all $u \in S-\{v\}$; note that this is exactly Equation (\ref{eqn:fsv}) in the case that $|S|>1$.

To establish the correctness of Equation (\ref{eqn:fstar}) recall that, by Definition \ref{def:vtxtour}, the arrival time of any $(s,1,X)$-tour in $\mathcal{G}$ is equal to the timestep after the timestep in which it traverses a time edge to reach the final unvisited vertex of $X$ for the first time. Assume that $I$ is a yes-instance and let $x^* \in X$ be the $k$-th unique vertex in $X$ that is visited by some foremost $(s,1,X)$-tour $W$; then, by the analysis in the previous paragraph, we must have $\alpha(W) = F(X,x^*)$ since $W$ is foremost, so $x^* = \arg\min_{v \in X} F(X,v)$ and thus $\alpha(W) = F(X,x^*) = \min_{v \in X} F(X,v) = F^*$, as required.

The fact that the answer returned by the algorithm is correct follows from the correctness of Equations (\ref{eqn:fsv}) and (\ref{eqn:fstar}) and the traceback procedure, together with the fact that $I$ is a no-instance if and only if $F^* = \infty$. The details of this second claim are not difficult to see and are omitted,
but we note that it is indeed possible that $F^*=\infty$ since $F^*$ is the summation of a number of values $sp(u,v,t)$, some of which may satisfy $sp(u,v,t)=\infty$ by definition.

\paragraph{Runtime analysis} Since we only compute values of $F(S,v)$ such that $v \in S$ and $1 \leq |S| \leq k$, in total we compute $O(\sum_{i=1}^k{k \choose i}i) = O(2^kk)$ values. Note that, to compute any value $F(S,v)$ with $|S| = i > 1$, Equation (\ref{eqn:fsv}) requires that we consider the values $F(S-\{v\},u) + sp(u,v,F(S-\{v\},u))$ with $u \in S-\{v\}$, of which there are exactly $i-1$.
We therefore use Theorem \ref{thm:shortpaths} to compute (and store temporarily), for each $S'$ with $|S'| = i-1$ and $x \in S'$, in $O(Ln^2)$ time the value of $sp(x,y,F(S',x))$ for all $y \in V(\mathcal{G})$ immediately after computing all $F(S',x)$, and use these precomputed shortest walk durations to compute $F(S,v)$ for any $S$ with $|S| = i$ and $v \in S$ in time $O(i) = O(k)$.
Thus, we spend $O(k) + O(Ln^2)  = O(Ln^2)$ (since $k \leq n$) time for each of $O(2^kk)$ values $F(S,v)$. This yields an overall time of $O(2^kkLn^2)$. Note that $F^*$ can be computed using Equation (\ref{eqn:fstar}) in $O(k)$ time since we take the minimum of $O(k)$ values; also note that a $(v,1,X)$-tour with arrival time $F^*$ can be reconstructed in time $O(kLn^2)$ using the aforedescribed traceback procedure, since we need to recompute $O(k)$ shortest walks, spending $O(Ln^2)$ time on each walk. Hence the overall running time of the algorithm is bounded by $O(2^kkLn^2)$, as claimed.
	\end{proof}

We remark that \textsc{$k$-fixed TEXP} is also in \textsf{FPT} when parameterized by the
lifetime~$L$: If $L<k-1$, the instance is clearly a no-instance, and if $L\ge k-1$,
the \textsf{FPT} algorithm for \textsc{$k$-fixed TEXP}
with parameter $k$ is also an \textsf{FPT} algorithm for parameter~$L$.

As \textsc{$k$-fixed TEXP} becomes \textsc{TEXP} when $X=V(\mathcal{G})$, we
get the following corollary.

\begin{corollary}
\label{cor:texplFPT}%
\textsc{TEXP} is in \textsf{FPT} when parameterized by the number of vertices~$n$ or by
the lifetime~$L$.
\end{corollary}

\subsection{\textup{\textsf{FPT}} algorithms for \texorpdfstring{\textsc{$k$-arbitrary TEXP}}{k-arbitrary TEXP}}\label{ss:karb}
The main result of this section is a randomized \textsf{FPT} algorithm for \textsc{$k$-arbitrary TEXP} that utilises the \textit{colour-coding} technique originally presented by Alon et al.~\cite{AYZ_95}. There, they employed the technique primarily to detect the existence of a $k$-vertex simple path in a given undirected graph $G$. More generally, it has proven useful as a technique for finding fixed motifs (i.e., prespecified subgraphs) in static graphs/networks. We provide a high-level description of the technique and the way that we apply it at the beginning of Section \ref{sec:karbrand}. A standard derandomization technique (originating from~\cite{AYZ_95,NSS_95}) is then utilised within Section~\ref{sec:karbdet} to obtain a deterministic algorithm for \textsc{$k$-arbitrary TEXP} with a worse, but still \textsf{FPT}, running time.

\subsubsection{A randomized algorithm}\label{sec:karbrand} 
The algorithm of this section employs the colour-coding technique of Alon et al.~\cite{AYZ_95}. First, we informally sketch the structure of the algorithm behind Theorem \ref{thm:karbFPT}: We colour the vertices of an input temporal graph uniformly at random, then by means of a DP subroutine we look for a temporal walk that begins at some start vertex $s$ in timestep $1$ and visits $k$ vertices with distinct colours by the earliest time possible. Notice that if such a walk is found then it must be a $(v,t,k)$-tour, since the $k$ vertices are distinctly coloured and therefore must be distinct. Then, the idea is to repeatedly (1) randomly colour the input graph $\mathcal{G}$'s vertices; then (2) run the DP subroutine on each coloured version of $\mathcal{G}$. We repeat these steps enough times to ensure that, with high probability, the vertices of an optimal $(s,1,k)$-tour are coloured with distinct colours at least once over all colourings -- if this happens then the DP subroutine will surely return an optimal $(s,1,k)$-tour. With this high-level description in mind, we now present/analyse the algorithm:
\begin{theorem}\label{thm:karbFPT}
	For every $\varepsilon > 0$, there exists a Monte Carlo algorithm that, with probability $1-\varepsilon$, decides a given instance $I = (\mathcal{G},s,k)$ of \textsc{$k$-arbitrary TEXP}, and returns an optimal solution if $I$ is a yes-instance, in time $O((2e)^kLn^3\log \frac{1}{\varepsilon})$, where $n = |V(\mathcal{G})|$ and $L$ is $\mathcal{G}'s$ lifetime.
\end{theorem}
\begin{proof}
	Let $V := V(\mathcal{G})$. We now describe our algorithm before proving it correct and analysing its running time. Let $c : V \to [k]$ be a colouring of the vertices $v \in V$. Let a walk $W$ in $\mathcal{G}$ that starts at $s$ and visits a vertex coloured with each colour in $D \subseteq [k]$ be known as a \textit{$D$-colourful walk}; let the timestep after the timestep at the end of which $W$
has for the first time visited vertices with $k$ distinct colours
be known as the \textit{arrival time} of~$W$, denoted by~$\alpha(W)$. The algorithm employs a subroutine that computes, should one exist, a $[k]$-colourful walk $W$ in $\mathcal{G}$ with earliest arrival time. Note that a $D$-colourful walk ($D \subseteq [k]$) in $\mathcal{G}$ is by definition an $(s,1,|D|)$-tour in $\mathcal{G}$. 
	
	Define $H(D,v)$ to be the earliest arrival time of any $D$-colourful walk (where $D \subseteq [k]$) in $\mathcal{G}$  that ends at a vertex $v$ with $c(v) \in D$.
 The value of $H(D,v)$ for any $D \subseteq [k]$ and $v$ with $c(v) \in D$ can be computed via the following dynamic programming formula (within the formula we denote by $D^-_{c(v)}$ the set $D - \{c(v)\}$):
	\begin{eqnarray}\label{eqn:hdv}
	H(D,v) = 
     			\begin{cases}
       				1+sp(s,v,1)		 &\; (|D| = 1)\\
       				\displaystyle \min_{u \in V: c(u) \in D^-_{c(v)}} [ H(D^-_{c(v)},u) + sp(u,v,H(D^-_{c(v)},u))] &\; (|D| > 1)\\
     			\end{cases}
	\end{eqnarray}
	In order to compute $H(D,v)$ for any $D \subseteq [k]$ and vertex $v$ with $c(v) \in D$, Equation (\ref{eqn:hdv}) requires that we consider values $H(D-\{c(v)\},u)$ such that $c(u) \in D-\{c(v)\}$, and so we begin by computing $H(D',v)$ for all $D'$ with $|D'| = 1$ and $v$ with $c(v) \in D'$, then for all $D'$ with $|D'| = 2$ and $v$ with $c(v) \in D'$, and so on, until all values $H([k],v)$ have been obtained. The earliest arrival time of any $[k]$-colourful walk in $\mathcal{G}$ is then given by
	\begin{eqnarray}\label{eqn:hdstar}
		H^* = \min_{u \in V(\mathcal{G})} H([k],u).
	\end{eqnarray}
	Once $H^*$ has been computed, we check whether its value is finite or equal to $\infty$. If $H^*$ is finite then we can use a pointer system and traceback procedure (almost identical to those used in the proof of Theorem \ref{thm:kfixedFPT}) to reconstruct an $(s,1,k)$-tour with arrival time $H^*$ if one exists; otherwise we return no. This concludes the description of the dynamic programming subroutine.
	
	Let $r = \lceil\frac{1}{\varepsilon}\rceil$ and let $W^*$ initially be the trivial walk that starts and finishes at vertex $s$ in timestep 1. Perform the following two steps for $e^k\ln r$ iterations:
	 
	\begin{enumerate} 
		\item Assign colours in $[k]$ to the vertices of $V$ uniformly at random and check if all $k$ colours colour at least one vertex of $G$; if not, start next iteration. If yes, proceed to step \ref{alg:randalg2}.\vspace{5pt}\label{alg:randalg1}
		\item Run the DP subroutine in order to find an optimal $[k]$-colourful walk $W$ in $\mathcal{G}$ if one exists. If such a $W$ is found then check if $\alpha(W) < \alpha(W^*)$ or $W^*$ starts and ends at $s$ in timestep $1$ (i.e., still has its initial value), and in either case set $W^* = W$; otherwise the DP subroutine returned no and we make no change to $W^*$.\label{alg:randalg2}
	\end{enumerate}
	
	Once all iterations of the above steps are over, check if $W^*$ is still equal to the walk that starts and finishes at $s$ in timestep $1$; if not then return $W^*$, otherwise return no. This concludes the algorithm's description. 	

	\paragraph{Correctness} We focus on proving the randomized aspect of the algorithm correct and omit correctness proofs for Equations (\ref{eqn:hdv}) and (\ref{eqn:hdstar}) since the arguments are similar to those provided in Theorem \ref{thm:kfixedFPT}'s proof.

	If $I$ is a no-instance then in no iteration will the DP subroutine find an $(s,1,k)$-tour in $\mathcal{G}$. Hence in the final step the algorithm will find that $W^*$ is equal to the walk that starts and ends at $s$ in timestep $1$ (by the correctness of Equations (\ref{eqn:hdv}) and (\ref{eqn:hdstar})) and return no, which is clearly correct. Assume then that $I$ is yes-instance. Let $W$ be an $(s,1,k)$-tour in $\mathcal{G}$ with earliest arrival time, and let $X \subseteq V$ be the set of $k$ vertices visited by $W$. Then, if during one of the $e^k\ln r$ iterations of steps \ref{alg:randalg1} and \ref{alg:randalg2} we colour the vertices of $V$ in such a way that $X$ is well-coloured (we say that a set of vertices $U \subseteq V$ is \textit{well-coloured} by colouring $c$ if $c(u) \neq c(v)$ for every pair of vertices $u,v \in U$), 
$W$ will induce an optimal $[k]$-colourful walk in $\mathcal{G}$. The DP subroutine will then return $W$ or some other optimal $[k]$-colourful walk $W'$ with $\alpha(W) = \alpha(W')$ that visits a well-coloured subset of vertices $X'$; note that the arrival time of the best tour found in any iteration so far will then surely be $\alpha(W)$, since $W$ has earliest arrival time. 
	
	Observe that if we colour the vertices of $V$ with $k$ colours uniformly at random, then, since $|X| = k$, there are $k^k$ ways to colour the vertices in $X \subseteq V$, of which $k!$ constitute well-colourings of $X$. Hence after a single colouring of $V$ we have
	\[\Pr[X \text{ is well-coloured}] = \frac{k!}{k^k} > \frac{1}{e^k},\]
where the inequality follows from the fact that $k!/k^k > \sqrt{2\pi}k^\frac{1}{2}e^\frac{1}{12k+1}/e^k$ (this inequality is due to Robbins~\cite{Robbins_55} and is related to Stirling's formula). Hence, after $e^k \ln r$ colourings, we have (using the standard inequality $(1-\frac{1}{x})^x \le \frac{1}{e}$ for all $x\ge 1$):
	\[\Pr[X \text{ is not well-coloured in any colouring}] \leq  \left(1-\frac{1}{e^k}\right)^{e^k\ln r} \leq 1/r \leq \varepsilon.\]
Thus, the probability that $X$
is well-coloured at least once after $e^k\ln r$ colourings is at least $1-\varepsilon$. It follows that, with probability $\geq 1-\varepsilon$, the earliest arrival $[k]$-colourful walk returned by the algorithm after all iterations is in fact an optimal $(s,1,k)$-tour in $\mathcal{G}$, since either $W$ or some other $(s,1,k)$-tour with equal arrival time will eventually be returned.
	
	\paragraph{Runtime analysis} Note that the DP subroutine computes exactly the values $H(D,v)$ such that $D \subseteq [k]$ and $v$ satisfies $c(v) \in D$. Hence there are at most ${k \choose i}n$ values $H(D,v)$ such that $|D| = i$, for all $i \in [k]$; this gives a total of $\sum_{i \in [k]} {k \choose i}n = O(2^kn)$ values. In order to compute $H(D,v)$ for any $D$ with $|D| = i > 1$, Equation~(\ref{eqn:hdv}) requires us to consider the value of $H(D-\{c(v)\},u) + sp(u,v,H(D-\{c(v)\},u))$ for all $u$ such that $c(u) \in D-\{c(v)\}$.
Therefore, similar to the algorithm in the proof of Theorem~\ref{thm:kfixedFPT}, we compute and store, immediately after computing each value $H(D',x)$ with $|D'| = i-1$ and $c(x) \in D'$, the value of $sp(x,y,H(D',x))$ for all $y\in V(\mathcal{G})$ in $O(Ln^2)$ time (Theorem~\ref{thm:shortpaths}).
Note that there can be at most $n$ vertices $u$ such that $c(u) \in D-\{c(v)\}$, and so in total we spend $O(n) + O(Ln^2) = O(Ln^2)$ time on each of $O(2^kn)$ values of $H(D,v)$, giving an overall time of $O(2^kLn^3)$. We can compute $H^*$ in $O(n)$ time since we take the minimum of $O(n)$ values, and the traceback procedure can be performed in $O(kLn^2)=O(Ln^3)$ time since we concatenate $k$ walks obtained using Theorem~\ref{thm:shortpaths}. Thus the overall time spent carrying out one execution of the DP subroutine is $O(2^kLn^3)$.

	Since the running time of each iteration of the main algorithm is dominated by the running time of the DP subroutine and there are $e^k\ln r = O(e^k \log \frac{1}{\varepsilon})$ iterations in total, we conclude that the overall running time of the algorithm is $O((2e)^kLn^3\log \frac{1}{\varepsilon})$, as claimed. This completes the proof.
\end{proof}

\subsubsection{Derandomizing the algorithm of Theorem \ref{thm:karbFPT}}\label{sec:karbdet}
The randomized colour-coding algorithm of Theorem \ref{thm:karbFPT} can be derandomized at the expense of incurring a $k^{O(\log k)}\log n$ factor in the running time. We employ a standard derandomization technique, presented initially in~\cite{AYZ_95}, which involves the enumeration of a \textit{$k$-perfect family of hash functions} from $[n]$ to $[k]$. The functions in such a family will be viewed as colourings of the vertex set of the temporal graph given as input to the \textsc{$k$-arbitrary TEXP} problem. 

Formally, a family $\mathcal{H}$ of hash functions from $[n]$ to $[k]$ is \textit{$k$-perfect} if, for every subset $S \subseteq [n]$ with $|S| = k$, there exists a function $f \in \mathcal{H}$ such that $f$ restricted to $S$ is bijective (i.e., one-to-one). The following theorem of Naor et al.\ \cite{NSS_95} enables one to construct such a family $\mathcal{H}$ in time linear in the size of $\mathcal{H}$:
\begin{theorem}[Naor, Schulman and Srinivasan~\cite{NSS_95}]\label{thm:kperfhf}
	A $k$-perfect family $\mathcal{H}$ of hash functions $f_i$ from $[n]$ to $[k]$, with size $e^kk^{O(\log k)}\log n$, can be computed in $e^kk^{O(\log k)}\log n$ time.
\end{theorem}
We note that the value of $f_i(x)$ for any $f_i \in \mathcal{H}$ and $x \in [n]$ can be evaluated in $O(1)$ time.

To solve an instance of $\textsc{$k$-arbitrary TEXP}$, we can now use the algorithm from the proof of
Theorem~\ref{thm:karbFPT}, but instead of iterating over $e^k\ln r$ random colourings, we iterate over
the $e^kk^{O(\log k)}\log n$ hash functions in the $k$-perfect family of hash functions constructed
using Theorem~\ref{thm:kperfhf}. This ensures that the set $X$ of $k$ vertices visited by an optimal
$(s,1,k)$-tour is well-coloured in at least one iteration, and we obtain the following theorem.
	\begin{theorem}\label{thm:karbFPTdet}
		There is a deterministic algorithm that can solve a given instance $(\mathcal{G},s,k)$ of $\textsc{$k$-arbitrary TEXP}$ in $(2e)^kk^{O(\log k)}Ln^3\log n$ time, where $n = |V(\mathcal{G})|$.
		If the instance is a yes-instance, the algorithm also returns an optimal solution.
	\end{theorem}

We remark that, since a temporal walk can visit at most $L+1$ vertices in a temporal graph with lifetime~$L$,
Theorem~\ref{thm:karbFPTdet} also implies an \textsf{FPT} algorithm for the following
problem, parameterized by the lifetime $L$ of the given temporal graph: Find a temporal walk that
visits as many distinct vertices as possible.

\subsection{\texorpdfstring{$\textup{\textsf{W}}[2]$-hardness of \textsc{Set TEXP} for parameter $L$}{W[2]-hardness
of Set TEXP for parameter L}}
\label{ss:set}%
The \textsf{NP}-complete \textsc{Hitting Set} problem is defined as
follows~\cite{Garey-Johnson/79}.

\begin{definition}[\textsc{Hitting Set}]
	An instance of \textsc{Hitting Set} is given as a tuple $(U, \mathcal{S}, k)$, 
	where $U=\{a_1,\dots,a_n\}$ is the ground set and $\mathcal{S} = \{S_1,\dots,S_m\}$ is 
	a set of subsets $S_i \subseteq U$. The problem then asks whether or not there exists a 
	subset $U' \subseteq U$ of size at most $k$ such that, for 
	all $i \in [m]$, there exists an $u \in U'$ such that $u \in S_i$.
\end{definition}
It is known that \textsc{Hitting Set} is $\textsf{W}[2]$-hard when parameterized by $k$~\cite{DF_99}.

\begin{theorem}
	\label{thm:sethard}%
	\textsc{Set TEXP} parameterized by $L$ (the lifetime of the 
	input temporal graph) is $\textup{\textsf{W}}[2]$-hard.
\end{theorem}

\begin{proof}
We give a parameterized reduction from the \textsc{Hitting Set} problem
with parameter $k$ to the \textsc{Set TEXP} problem with parameter~$L$.
Given an instance $I=(U,\mathcal{S},k)$ of \textsc{Hitting Set}, we construct
an instance $I'=(\mathcal{G},s,\mathcal{X})$ of \textsc{Set TEXP} as follows:
The lifetime of $\mathcal{G}$ is set to $L=k$.
In each of the $L$ steps, the graph is a complete graph with vertex
set $U \cup \{s\}$, where $s$ is a start vertex that is assumed not to be in~$U$.
Finally, we set $\mathcal{X}=\mathcal{S}$. We proceed to show
that $I$ is yes-instance if and only if $I'$ is a yes-instance.

If $I$ is a yes-instance, let $U'=\{u_1,u_2,\ldots,u_k\}$ be a hitting
set of size~$k$. Then the walk that moves from $s$ to $u_1$ in step~$1$
and then from $u_{i-1}$ to $u_i$ in step $i$ for $2\le i\le k$ 
is an $(s,1,U')$-tour that visits at least one vertex from each
set in $\mathcal{X}$. Therefore, $I'$ is a yes-instance.

If $I'$ is a yes-instance, let $W$ be a strict temporal
walk that visits at least one vertex from each set in~$\mathcal{X}$.
Let $U'$ be the set of at most $L=k$ vertices that this walk visits
in addition to the start vertex~$s$. Then $U'$ is a hitting set
for $I$. Hence, $I$ is a yes-instance.
\end{proof}

Similar to the case of \textsc{$k$-fixed TEXP},
we can remark that \textsc{$k$-arbitrary TEXP} is also in \textsf{FPT} when parameterized by the
lifetime~$L$: If $L<k-1$, the instance is clearly a no-instance, and if $L\ge k-1$,
the \textsf{FPT} algorithm for \textsc{$k$-arbitrary TEXP}
with parameter $k$ is also an \textsf{FPT} algorithm for parameter~$L$.

\section{Non-Strict \textsc{TEXP} parameterizations}\label{sec:nstexp}
In this section, we study temporal exploration problems in the
non-strict setting.
Let $\mathcal{G} = \langle G_1,\ldots,G_L\rangle$ be the given non-strict temporal graph,
and let $s\in V(\mathcal{G})$ be the given start vertex.
When analysing running-times in this section, we assume that the non-strict
temporal graph is given by providing, for each timestep~$t$, a list of the vertex sets (with
each of these sets given as a list of vertices) of the components in that timestep.
This representation has size $\Theta(Ln)$. If the graph was given in the same form as
a strict temporal graph, this representation could be computed by a pre-processing
step that runs in time $O(Ln^2)$.

First, we show in Section~\ref{ss:kfixNS} that \textsf{FPT} algorithms for \textsc{$k$-fixed NS-TEXP} and \textsc{$k$-arbitrary NS-TEXP}
can be derived using similar techniques as in Section~\ref{sec:stricttexp}.
After that, we show that \textsc{NS-TEXP} and its variants can all be solved in polynomial
time for $\gamma=2$ (Section~\ref{ss:gammaTwo}) and that
\textsc{NS-TEXP} is in \textsf{FPT} when 
parameterized by the lifetime $L$ (Section~\ref{ss:NStexpFPT}).
Finally, we prove $\textsf{W}[2]$-hardness for the \textsc{Set NS-TEXP} problem when the same parameter is considered (Section~\ref{ss:NSsetWTwo}).

\subsection{\texorpdfstring{\textsc{$k$-fixed NS-TEXP} and \textsc{$k$-arbitrary NS-TEXP}}{k-fixed NS-TEXP and k-arbitrary NS-TEXP}}
\label{ss:kfixNS}%
We now define $sp(u,v,t)$
as the duration of a shortest (i.e., having minimum arrival time) \emph{non-strict}
temporal walk in $\mathcal{G}$ that starts at $u \in V(\mathcal{G})$ in timestep $t$ and ends at $v \in V(\mathcal{G})$. If $u=v$ or if $u$ and $v$ are in the same component in step~$t$, then
$sp(u,v,t)=0$. If there is no such non-strict temporal walk, we let $sp(u,v,t)=\infty$.

\begin{lemma}
\label{lem:nonstrictsp}%
For given $u$ and $t$, one can compute the
values $sp(u,v,t)$ for all $v\in V(\mathcal{G})$ in $O(Ln)$ time. Furthermore, for
each $v\in V(\mathcal{G})$, a shortest walk starting at $u$ at time $t$ and reaching $v$
can then be determined in time proportional to $1+sp(u,v,t)$.
\end{lemma}

\begin{proof}
Let $V=V(\mathcal{G})$.
For each $w\in V$, maintain a label $r(w)$ to represent whether $w$ is reachable
by the time step under consideration, and a label $a(w)$ to represent the
earliest arrival time at $w$ if $w$ is reachable. In addition, we will remember
a predecessor $p(w)$ for every reachable vertex.
Initialise the current time to $t_c=t$; set $r(w)=\True$, $a(w)=t_c$ and $p(w)=u$ for
all $w$ in the component of $u$ at time~$t_c$; set $r(w)=\False$ and $a(w)=\infty$
for all other vertices. This takes $O(n)$ time.

Then repeat the following step until either all vertices are reachable or $t_c$ equals
the lifetime of the graph:
Increase $t_c$ by one. For each component $B$ of step $t_c$, check whether $B$ contains
a vertex $w$ with $r(w)=\True$ and, if so, mark~$B$ and remember $w$ as $p_B$.
For each vertex $w$ with $r(w)=\False$
in any marked component $B$ of step $t_c$, we then set $r(w)=\True$, $a(w)=t_c$
and $p(w)=p_B$.
Each execution of this step takes $O(n)$ time.

Finally, for each vertex $v\in V$, we set $sp(u,v,t)=a(v)-t$.

To construct the shortest temporal walk corresponding to a value $sp(u,v,t)$,
we trace back the vertices (and their components) starting with $v$ (visited
at time $t'=t+sp(u,v,t)$), $p(v)$ (visited at time $a(p(v))\le t'-1$),
$p(p(v))$, and so on.

It is clear that the running-time is $O(Ln)$. Correctness can be shown by induction:
When the step for value $t_c$ has been completed, a vertex $w$ satisfies $r(w)=\True$
if and only if $w$ is reachable from $u$ with arrival time at most $t_c$, and
in that case $a(w)=t'$ is the earliest arrival time at $w$ and, if $t'>t$,
$p(w)$ is a vertex that is reachable with arrival time at most $t'-1$ and
from which $w$ can be reached in step $t'$.
\end{proof}

Next, we observe that it is easy to see that Equations~(\ref{eqn:fsv}) and (\ref{eqn:fstar}) from
the proof of Theorem~\ref{thm:kfixedFPT} remain valid in the non-strict case, as
the arguments for correctness remain the same. The factor $Ln^2$ in the running-time
of Theorem~\ref{thm:kfixedFPT} improves to $Ln$ in the non-strict case as,
by Lemma~\ref{lem:nonstrictsp}, it takes only $O(Ln)$ time to compute $sp(u,v,t)$
for all $v\in V$ right after $F(S',u)=t$ has been computed for some set $S'$
and $u\in S'$. Thus, we obtain:

\begin{corollary}\label{cor:NSkfixedFPT}
It is possible to decide any instance $I = (\mathcal{G},s,X,k)$ of \textsc{$k$-fixed NS-TEXP}, and return an optimal solution if $I$ is a yes-instance, in time $O(2^kkLn)$, where $n = |V(\mathcal{G})|$ and $L$ is $\mathcal{G}$'s lifetime.
\end{corollary}

Similarly, Equations (\ref{eqn:hdv}) and (\ref{eqn:hdstar}) from the proof of
Theorem~\ref{thm:karbFPT} remain valid, and the derandomization used in the
proof of Theorem~\ref{thm:karbFPTdet} works for the non-strict case without
any alterations. Thus, we obtain the following corollary of
Theorems~\ref{thm:karbFPT} and~\ref{thm:karbFPTdet}, where again we save
a factor of $n$ in the running-time because we can use Lemma~\ref{lem:nonstrictsp}
instead of Theorem~\ref{thm:shortpaths}.

\begin{corollary}\label{cor:NSkarbFPT}
	For every $\varepsilon > 0$, there exists a Monte Carlo algorithm that, with probability $1-\varepsilon$, decides a given instance $I = (\mathcal{G},s,k)$ of \textsc{$k$-arbitrary NS-TEXP}, and returns an optimal solution if $I$ is a yes-instance, in time $O((2e)^kLn^2\log \frac{1}{\varepsilon})$, where $n = |V(\mathcal{G})|$ and $L$ is $\mathcal{G}'s$ lifetime.
Furthermore, there is a deterministic algorithm that can solve a given instance $(\mathcal{G},s,k)$ of $\textsc{$k$-arbitrary NS-TEXP}$ in $(2e)^kk^{O(\log k)}Ln^2\log n$ time.
If the instance is a yes-instance, the algorithm also returns an optimal solution.
\end{corollary}
\subsection{Non-strict exploration with at most two components per step}
\label{ss:gammaTwo}%
Let $\mathcal{G} = \langle G_1,\ldots,G_L\rangle$ be the given non-strict temporal graph.
If there is a step $t$ in which the partition $G_t$ consists of a single
component $C_{t,1}$, then it it trivially possible to visit all vertices:
We simply wait at the start vertex until step~$t$, and then visit all
vertices in step~$t$.
Therefore, for all four problem variants (\textsc{NS-TEXP},
\textsc{$k$-fixed NS-TEXP}, \textsc{$k$-arbitrary NS-TEXP}, and
\textsc{Set NS-TEXP}),
instances where the maximum number of components per
step is $\gamma=1$ are trivially yes-instances, and instances with $\gamma=2$
are also yes-instances if at least one step has a single component.
In the remainder of this section, we therefore consider the case $\gamma=2$
under the assumption that the partition in every step consists of exactly two
components. Furthermore, we can assume without loss of generality that
no two consecutive steps have the same two components: Any number
of consecutive steps that all have the same two components could be
replaced by a single step without changing
the answer to any of the four variants of the \textsc{NS-TEXP} problem.

First, we are interested in the movements that the partitions in two consecutive steps
allow. We refer to two consecutive steps $i$ and $i+1$ as a \emph{transition}.

\begin{definition}
\label{def:transitions}%
A transition between step $i$ with partition $G_i=(A_i,B_i)$
and step $i+1$ with partition $G_{i+1}=(A_{i+1},B_{i+1})$
is called \emph{free} if the four sets $A_i\cap A_{i+1}$,
$A_i\cap B_{i+1}$, $B_i\cap A_{i+1}$, $B_i\cap B_{i+1}$
are all non-empty. If exactly one of these sets
is empty, the transition is called \emph{restricted}.
\end{definition}

\begin{figure}
\centering
\scalebox{0.65}{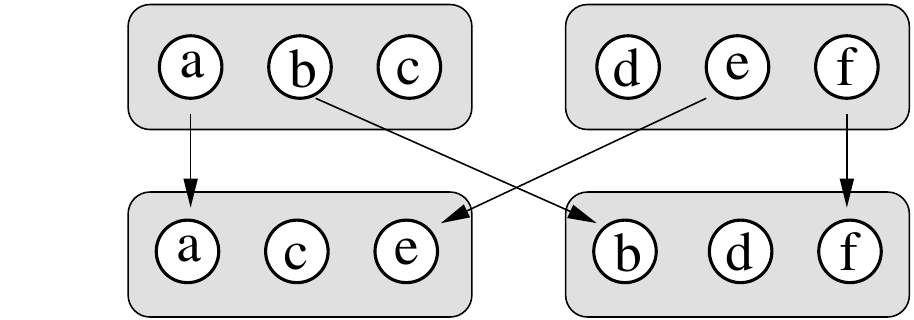}
\hspace{1cm}
\scalebox{0.65}{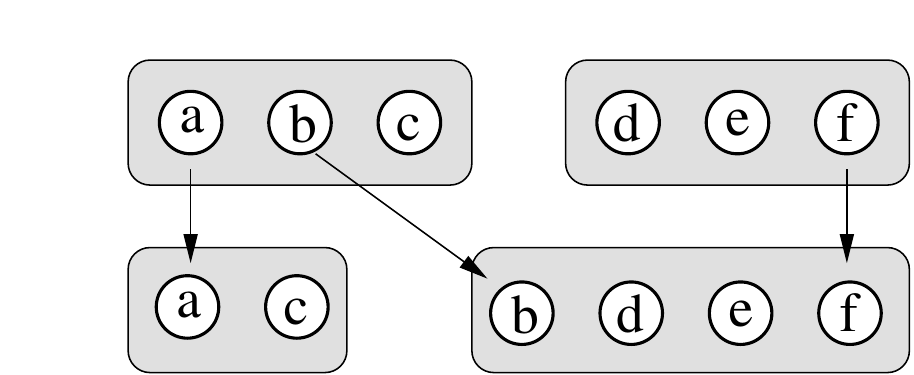}

\caption{Free transition (left) and restricted transition (right).}
\label{fig:freerest}
\end{figure}%
See Figure~\ref{fig:freerest} for an illustration.
In a free transitions, a walk can reach any of the two
components in step $i+1$ no matter which component the
walk visits in step~$i$. In a restricted transition,
there is one component in step~$i$ such that one component
in step~$i+1$ cannot be reached from it. We show next
that these are the only possible types of transitions.

\begin{lemma}
\label{lem:freeorrest}%
Every transition is either free or restricted.
\end{lemma}

\begin{proof}
Assume that the transition from $G_i$ to $G_{i+1}$ is
not free. Assume without loss of generality that
$A_i\cap B_{i+1}$ is empty.
This means that
every vertex of $A_i$ must be contained in $A_{i+1}$.
As we assume that the partitions of consecutive steps
are different, we get that $A_i\subset A_{i+1}$
and, hence, $B_i \supset B_{i+1}$. This implies
that at least one vertex from $B_i$ is in $A_{i+1}$.
Furthermore, neither $A_{i+1}$ nor $B_{i+1}$ can be
empty, so there must also be a vertex in $B_i\cap B_{i+1}$.
Hence, the transition is restricted.
\end{proof}

The proof of Lemma~\ref{lem:freeorrest} shows
that in a restricted transition there is one component
that shrinks (gets replaced by a strict subset)
and one that grows (gets replaced by a strict superset).
We call the former the \emph{shrinking component}
and the latter the \emph{growing component} (as
indicated in Figure~\ref{fig:freerest}).

\begin{lemma}
\label{lem:shrinkrest}%
If there is a restricted transition from step $i$ to $i+1$, a walk
that visits the shrinking component in step~$i$
can visit all vertices of the graph in steps $i$ and $i+1$.
\end{lemma}

\begin{proof}
The walk can visit all vertices of the
shrinking component in step $i$ and then end step $i$ at
a vertex that leaves the shrinking component. In step
$i+1$, the walk then visits all vertices in the component that
has grown. It is easy to see that every vertex is contained
in the two components visited by the walk.
\end{proof}

\begin{lemma}
\label{lem:freerestall}
If a restricted transition follows a free transition,
the whole graph can be explored.
\end{lemma}

\begin{proof}
Assume that there is a free transition from step $i-1$
to step $i$ and a restricted transition from step $i$
to step $i+1$. Let $B_i$ be the shrinking component
in the restricted transition. Then a walk can visit
$B_i$ in step $i$ (because the free transition allows
it to reach $B_i$) and then, by Lemma~\ref{lem:shrinkrest},
visit all remaining unvisited vertices in step~$i+1$.
\end{proof}

\begin{lemma}
\label{lem:logfree}%
In $1+\log_2 n$ consecutive free transitions,
the whole graph can be explored.
\end{lemma}

\begin{proof}
Let $A$ be the component that the walk visits in the
first step of the first free transition. In each of
the $1+\log_2 n$ free transitions, we can choose as component
to visit in the next step the one that contains more of the
previously unvisited vertices. In this way, we are guaranteed
to visit at least half of all the remaining unvisited vertices in each
of these $1+\log_2 n$ steps. The number of unvisited vertices
remaining at the end of these $1+\log_2 n$ steps is hence
at most $n/2^{1+\log_2 n}<1$.
\end{proof}

\begin{theorem}
\label{thm:gammaTwo}%
There is an algorithm that solves instances of \textsc{NS-TEXP}
with $\gamma=2$ in $O(Ln+n^2\log n)$ time.
\end{theorem}

\begin{proof}
In $O(Ln)$ time, we can check whether there is a step
in which there is a single component (in that case,
we output ``yes'' and terminate). In the same time bound,
we also preprocess the graph to ensure that no two consecutive
steps have the same partition and determine for each transition
whether it is free or restricted.

If a restricted transition follows a free transition,
we can output ``yes'' by Lemma~\ref{lem:freerestall}.
Otherwise, there must be an initial (possibly empty)
sequence $\mathcal{R}$ of restricted transitions, followed by a (possibly empty)
sequence $\mathcal{F}$ of free transitions.

If the start vertex $s$ is in the shrinking component
in one of the restricted transitions $\mathcal{R}$, then
we can visit all vertices of the graph by Lemma~\ref{lem:shrinkrest},
so we output ``yes''.
Otherwise, the start vertex $s$ must be
in the growing component in
all the restricted transitions $\mathcal{R}$. In this case,
it is impossible to leave that component. No decision
needs to be made during $\mathcal{R}$, and the walk
must visit the component containing $s$ in the first
time step of the first free transition.

If the number of free transitions in $\mathcal{S}$ is greater than $1+\log_2 n$,
the answer is ``yes'' by Lemma~\ref{lem:logfree}.
Otherwise, there are at most $1+\log_2 n$ free transitions.
Then, all possible choices for the next component to visit during
each of the at most $1+\log_2 n$ free transitions can be enumerated
in $O(2^{1+\log_2 n})=O(n)$ time. Furthermore, for each of these possibilities,
one can check in $O(n\log n)$ time whether the corresponding walk
visits all vertices of the graph.
\end{proof}

\begin{corollary}
\label{cor:gammaTwo}%
For each of the problems \textsc{$k$-fixed NS-TEXP}, \textsc{$k$-arbitrary NS-TEXP},
and \textsc{Set NS-TEXP},
there is an algorithm that solves instances with $\gamma=2$ in $O(Ln+n^2\log n)$ time.
\end{corollary}

\begin{proof}
First, assume that
there is a step with a single component, or that a restricted transition follows
a free transition, or that the vertex $s$ is ever contained in the shrinking
component of a restricted transition, or that the number of free transitions
is greater than $1+\log_2 n$. In all these cases, as argued in the
proof of Theorem~\ref{thm:gammaTwo}, all vertices of the input graph
can be visited, and hence the given instance is a yes-instance also of the three
problem variants under consideration here.

Now, assume that the temporal graph consists of an initial (possibly empty)
sequence $\mathcal{R}$ of restricted transitions such that $s$ is always
contained in the growing component, followed by 
a sequence $\mathcal{F}$ of at most $1+\log_2 n$ free transitions.
Then there are at most $2^{1+\log_2 n}=O(n)$ possible non-strict
temporal walks in the graph, and we can simply enumerate them
all and check for each of them in $O(n \log n)$ time whether it is
a solution to the given variant of \textsc{NS-TEXP}.
\end{proof}

\subsection{An \textup{\textsf{FPT}} algorithm for \textsc{NS-TEXP} with parameter \texorpdfstring{$L$}{L}}
\label{ss:NStexpFPT}
We now consider $\textsc{NS-TEXP}$ parameterized by the lifetime $L$ of the input temporal graph $\mathcal{G}$.
Let an instance of $\textsc{NS-TEXP}$ be given as a tuple $(\mathcal{G},s,L)$. We prove that $\textsc{NS-TEXP}$ is
in $\textsf{FPT}$ for parameter $L$ by specifying a bounded search tree-based \textsf{FPT} algorithm.

Let $\mathcal{G} = \langle G_1,\ldots,G_L\rangle$ be some non-strict temporal graph. Throughout this section we let 
$\mathcal{C}(\mathcal{G}) := \bigcup_{t \in [L]} G_t$, i.e., $\mathcal{C}(\mathcal{G})$ is the set of all components belonging to some layer of $\mathcal{G}$.
We implicitly assume that each component $C\in \mathcal{C}(\mathcal{G})$ is \emph{associated} with a unique layer $G_t$ of $\mathcal{G}$ in which it is contained.
If a component (seen as just a set of vertices) occurs in several layers, we thus treat these occurrences as different
elements of $\mathcal{C}(\mathcal{G})$ (or of any subset thereof) because they are associated with different layers.
If $Q$ is a set of components in $\mathcal{C}(\mathcal{G})$ that are associated with distinct layers (i.e.,
no two components in $Q$ are associated with the same layer $G_t$ of $\mathcal{G}$), then we say that
the components in $Q$ \emph{originate from unique layers of $\mathcal{G}$}.
For a set $Q$ of components that originate from unique layers of $\mathcal{G}$,
we let $D(Q) := \bigcup_{C \in Q} C$ be the union of the vertex sets of the components in~$Q$.
For any such set $Q$, we also let $T(Q) = \{t \in [L] : \text{there is a }C \in Q \text{ associated with layer } G_t\}$.

Within the following, we assume that $\mathcal{G}$ admits a non-strict exploration schedule~$W$.
\begin{observation}\label{obs:unvisit}
	Let $Q$ ($|Q| \in [0,L-1]$) be a subset of the components visited by the exploration schedule~$W$. Then there exists $C \in \mathcal{C}(\mathcal{G}) - Q$ with $C \in G_{t}$ ($t \in [L] - T(Q)$) such that $|C - D(Q)| \geq (n-|D(Q)|)/(L-|T(Q)|)$.
\end{observation}
Observation~\ref{obs:unvisit} follows since, otherwise, $W$ visits at most $L-|T(Q)|$ components $C \in \mathcal{C}(\mathcal{G})-Q$ that each contain $|C - D(Q)| < (n-|D(Q)|)/(L-|T(Q)|)$ of the vertices $v \notin D(Q)$, and so the total number of vertices visited by $W$ is strictly less than~$|D(Q)| +  (L-|T(Q)|)\cdot(n - |D(Q)|)/(L-|T(Q)|) = n$, a contradiction.

We briefly outline the main idea of our \textsf{FPT} result: We use a search tree algorithm that maintains a set $Q$ of
components that a potential exploration schedule could visit, starting with the empty set. 
Then the algorithm repeatedly tries all possibilities for adding a component (from some so far untouched layer)
that contains at least $(n-|D(Q)|)/(L-|T(Q)|)$ unvisited vertices (whose existence is guaranteed by
Observation~\ref{obs:unvisit} if there exists an exploration schedule).
It is clear that the search tree has depth~$L$, and the main further ingredient is an argument
showing that the number of candidates for the component to be added is bounded by a function of $L$,
namely, by $(L-|T(Q)|)^2$: This is because each of the $L-|T(Q)|$ untouched layers
can contain at most $L-|T(Q)|$ components that each contain at least $(n-|D(Q)|)/(L-|T(Q)|)$ unvisited vertices.
We now proceed to describe the details of the algorithm and its analysis.
\begin{lemma}\label{lem:compreach}
Let $\mathcal{G} = \langle G_1,\ldots,G_L\rangle$ be an arbitrary order-$n$ non-strict temporal graph. Then, for components $C_{t_1,j_1} \in G_{t_1}$ and $C_{t_2,j_2} \in G_{t_2}$ (with $1 \leq t_1 \leq t_2 \leq L$) one can decide, in $O((t_2-t_1+1)n)$ time, whether there exists a non-strict temporal walk beginning at any vertex contained in $C_{t_1,j_1}$ in timestep~$t_1$ and finishing at $C_{t_2,j_2}$ in timestep~$t_2$.
\end{lemma}
\begin{proof}
	For any $v \in V(\mathcal{G})$ and $t \in [t_1,t_2]$, let $c(v,t)$ denote the component $C_{t,j}$ such that $v \in C_{t,j}$ during timestep $t$. First, precompute the values $c(v,t)$ by, for every $t \in [t_1,t_2]$, scanning each component $C \in G_t$ and setting $c(v,t) = C$ if and only if $v \in C$. Next, let $X_{t_1} = \{C_{t_1,j_1}\}$ and then consider the timesteps $t \in [t_1+1,t_2]$ in increasing order, constructing at each timestep $t$ the set ${X_t = X_{t-1} \cup \{c(v,t) : v \in \bigcup_{C \in X_{t-1}} C\}}$. Finally, check whether $C_{t_2,j_2} \in X_{t_2}$, returning \textsf{yes} if so and \textsf{no} otherwise.
	
	The correctness of the algorithm is not hard to see. To see that the claimed running time of $O((t_2-t_1+1)n)$ holds, note first that precomputing the values $c(v,t)$ for any $v \in V(\mathcal{G})$ and any $t \in [t_1,t_2]$ requires $O((t_2-t_1+1)n)$ time since, in each timestep $t \in [t_1,t_2]$, we simply iterate over the vertices (of which there are always $n$ in total) contained in each component $C \in G_t$ .  Then, to compute $X_t$ for each $t \in [t_1+1,t_2]$, we add $c(v,t)$ (which can be evaluated in $O(1)$ time due to our preprocessing step) to $X_t$ for each vertex $v \in \bigcup_{C \in X_{t-1}} C$, of which there can be at most $n$. This second step of the algorithm can clearly also be executed in $O((t_2-t_1+1)n)$ time, and the lemma follows.
\end{proof}
Let $Q$ 
be a set of components originating from unique layers of $\mathcal{G}$, and let $W^?_\mathcal{G}(s, Q) = \textsf{yes}$ if and only if there exists a non-strict temporal walk in $\mathcal{G}$ that starts at $s \in V(\mathcal{G})$ in timestep $1$ and visits at least the components contained in $Q$, and \textsf{no} otherwise.
\begin{lemma}\label{lem:wexpcorr}
	For any order-$n$ non-strict temporal graph $\mathcal{G} = \langle G_1,\ldots,G_L\rangle$, any $s \in V(\mathcal{G})$, and any set $Q$ of components originating from unique layers of $\mathcal{G}$, $W^?_\mathcal{G}(s, Q)$ can be computed in $O(Ln)$ time.
\end{lemma}
\begin{proof}
	Let $C_{s_1},C_{s_2},\ldots,C_{s_{|Q|}}$ be an an index-ordered sequence of the components in $Q$, with the indices $s_i \in [L]$ satisfying $C_{s_i} \in G_{s_i}$ (for all $i \in [|Q|]$) and $s_{i} < s_{i+1}$ (for all $i \in [|Q|-1]$). Let $C_{s} \in G_1$ be the unique component in layer $1$ such that $s \in C_{s}$ (note that we may have $C_{s_1} = C_{s}$). Now, apply the algorithm of Lemma~\ref{lem:compreach} with $C_{t_1,j_1} = C_s$ and $C_{t_2,j_2} = C_{s_1}$, and then with $C_{t_1,j_1} = C_{s_i}$ and $C_{t_2,j_2} = C_{s_{i+1}}$ for all $i \in [|Q|-1]$. If the return value of any application of the algorithm of Lemma~\ref{lem:compreach} is \textsf{no}, then we return $W^?_\mathcal{G}(s,Q) = \textsf{no}$; otherwise we return $W^?_\mathcal{G}(s,Q) = \textsf{yes}$. This concludes the algorithm's description.
	
Since each component $C_{s_i}$ can only be visited in timestep $s_i$ it is clear that
any walk that visits all components of $Q$ must visit them in the specified order.
The algorithm sets $W^?_\mathcal{G}(s,Q) = \textsf{yes}$ if the components of $Q$ can be visited in the specified order.
On the other hand, if the algorithm of Lemma~\ref{lem:compreach} returns \textsf{no} for at least one pair of input components $C_{s_i},C_{s_{i+1}}$ (or $C_s,C_{s_1}$), then it must be that the components cannot be visited in this order, and thus the algorithm sets $W^?_\mathcal{G}(s,Q) = \textsf{no}$. Thus, the algorithm's correctness follows from the correctness of Lemma~\ref{lem:compreach}'s algorithm. To see that the running-time of the algorithm is bounded by $O(Ln)$, recall that each application of Lemma~\ref{lem:compreach}'s algorithm to start/finish components $C_{s_i}$ and $C_{s_{i+1}}$ takes $c(s_{i+1}-s_i+1)n$ time (for a constant $c$ hidden in the bound of Lemma~\ref{lem:compreach}). Thus the total amount of time spent over all applications is $c(s_1 - 1 + 1)n + \sum_{i \in [|Q|-1]} c(s_{i+1}-s_i+1)n = cn(s_{|Q|} + |Q| - 1) \leq cn(2L - 1) = O(Ln)$, where the last inequality holds since $|Q|,s_{|Q|} \leq L$.
\end{proof}
Now, let $\mathcal{G}$ be some input graph, and let $Q$ be some set of components originating from unique layers of $\mathcal{G}$. For any $s \in V(\mathcal{G})$, the recursive function $g(\mathcal{G},s,Q)$ (Algorithm~\ref{alg:computeg}) returns \textsf{yes} if and only if there exists a non-strict exploration schedule of $\mathcal{G}$ that starts at $s$ and visits (at least) the components contained in $Q$, and returns \textsf{no} otherwise. We prove the correctness of Algorithm~\ref{alg:computeg} in Lemma~\ref{lem:computegcorrect}.
\begin{algorithm}
  \DontPrintSemicolon
  \SetAlgoLined
  \caption{Recursive function $g(\mathcal{G},s,Q)$.}\label{alg:computeg}
  \BlankLine
  	\eIf{$|Q| = L$ or $|D(Q)| = n$}
  	{
  		\lIf{$|D(Q)| = n$}{\Return{$W^?_\mathcal{G}(s,Q)$}}
  		\lElse{\Return{\textsf{no}}}
  	}
  	{
  	  	$C' \gets \{C \in \mathcal{C}(\mathcal{G}) - Q : |C - D(Q)| \geq (n-|D(Q)|)/(L-|T(Q)|)\}$\;
  	  	$C^* \gets C' - \{C \in C' : C \in G_t, t \in T(Q)\}$.\\
  	  	\lIf{$|C^*| = 0$}
  	  	{
  	  		\Return{\textsf{no}}
  	  	}
  	  	\For{$C \in C^*$}
  	  	{
  	  		\lIf{$g(\mathcal{G},s,Q\cup \{C\}) = \textsf{yes}$}{\Return{\textsf{yes}}}
  		}
  		\Return{\textsf{no}}
  	}
\end{algorithm}
\begin{lemma}\label{lem:computegcorrect}
	For any non-strict temporal graph $\mathcal{G}$, any $s \in V(\mathcal{G})$, and any set $Q$ (with $|Q| \in [0,L]$) containing components originating from unique layers of $\mathcal{G}$, Algorithm~\ref{alg:computeg} correctly computes $g(\mathcal{G},s,Q)$.
\end{lemma}
\begin{proof}
	We first show that $g(\mathcal{G},s,Q)$ is correct in the base case, i.e., when $|Q| = L$ or $|D(Q)| = n$. If we have $|D(Q)| = n$, then any non-strict temporal walk that starts at $s$ in timestep $1$ and visits all components in $Q$ is an exploration schedule. Thus, the correctness of line $2$ follows from the definition of the return value $W^?_\mathcal{G}(s,Q)$ (which can be computed using Lemma~\ref{lem:wexpcorr}). If $|Q| = L$ and $|D(Q)| < n$, i.e., we have reached line $3$, then there must exist no exploration schedule that visits each of the components in $Q$, since any non-strict temporal walk in a temporal graph with lifetime $L$ can visit at most $L$ components, but at least one additional component $C \notin Q$ needs to be visited to cover at least one vertex $v \notin D(Q)$ -- thus it is correct to return $\textsf{no}$ in this case.

Otherwise, we have $|Q| < L$ and $|D(Q)| < n$, and are in the recursive case. Then, by Observation~\ref{obs:unvisit}, any non-strict exploration schedule that visits all components in $Q$ must visit at least one other component $C \in \mathcal{C}(\mathcal{G}) - Q$ such that $|C - D(Q)| \geq (n-|D(Q)|)/(L-|T(Q)|)$. Line $5$ computes the set $C'$ consisting of all such components, line 6 forms from $C'$ the set $C^*$ by removing from $C'$ any components that originate from layers $G_t$ such that $C \in G_t$ for some $C \in Q$ (since only one component can be visited in each timestep, and thus we want $Q$ to be a set of components originating from unique layers of $\mathcal{G}$).
We remark that a more efficient implementation could skip layers $G_t$ with $t\in T(Q)$ already when constructing $C'$
in line~5, but the asymptotic running-time of the overall algorithm would not be affected by this change.
The correctness of line 7 follows from Observation~\ref{obs:unvisit}. To complete the proof, we claim that the value $\textsf{yes}$ is returned by line $9$ if and only if there exists a non-strict temporal exploration schedule starting at $s$ that visits all the components contained in $Q$; we proceed by reverse induction on $|Q|$. Assume first that the return value of $g(\mathcal{G},s,Q')$ is correct for any $Q'$ with $|Q'| = k$ ($k \in [L]$) and let $|Q| = k-1$. Now assume that, during the execution of $g(\mathcal{G},s,Q)$, line $9$ returns \textsf{yes}; it follows that $g(\mathcal{G},s,Q') = \textsf{yes}$ for some $Q' = Q \cup {C}$ with $C \in C^{*}$ and thus it follows from the induction hypothesis that there exists a non-strict temporal exploration schedule that starts at $s$ and visits all the components contained in~$Q$, as required. In the other direction, assume that there exists some non-strict exploration schedule $W$ that starts at $s$ in timestep $1$ and visits all the components in $Q$. Note that, since the execution has reached line 9, we surely have $|C^*| > 0$; since we also have $|Q| < L$ and $|D(Q)| < n$ it follows from Observation~\ref{obs:unvisit} that $W$ visits at least one additional component $C \in C^{*}$. Then, by the induction hypothesis, we must have $g(\mathcal{G},s,Q \cup \{C\}) = \textsf{yes}$; thus when the loop of lines 8--10 processes $C \in C^*$ the algorithm will return \textsf{yes} as required.
\end{proof}

\begin{theorem}\label{thm:nstexpl}
There is an algorithm that decides any instance $I = (\mathcal{G},s,L)$ of $\textsc{NS-TEXP}$
in $O(L(L!)^2n)$ time.
\end{theorem}

\begin{proof}
The algorithm simply returns the value of function call $g(\mathcal{G}, s, \emptyset)$ (Algorithm~\ref{alg:computeg}).

By Lemma~\ref{lem:computegcorrect}, $g(\mathcal{G},v,Q)$ returns $\textsf{yes}$ if and only if $\mathcal{G}$ admits a non-strict exploration schedule that starts at $v$ and visits at least the components contained in the set $Q$ (which contains $|Q| \in [0,L]$ components originating from unique layers of $\mathcal{G}$), and returns \textsf{no} otherwise. Thus the correctness of the above follows immediately.

In order to bound the running time of the above algorithm, it suffices to bound the running time of Algorithm~\ref{alg:computeg}, i.e., the  recursive function $g$. The initial call is $g(\mathcal{G},s,\emptyset)$, and each recursive call is of the form $g(\mathcal{G},s,Q)$ where $Q$ is a set of components with size one more than the input set of the parent call. Hence, line $1$ ensures that there are at most $L$ levels of recursion in total (not including the level containing the initial call). For a call at level $i \geq 0$, the set $C^{*}$ constructed in line 5 has size at most $(L-i)^2$, since at most $L-i$ components can cover at least $(n-|D(Q)|)/(L-i)$ of the vertices in $V(\mathcal{G}) - D(Q)$ during each of the $L-i$ steps $t \in [L] - T(Q)$. Thus each call at level $i \ge 0$ makes at most $(L-i)^2$ recursive calls. The tree of recursive calls thus has at most $(L!)^2$ nodes
at depth $L$, and hence $O((L!)^2)$ nodes in total. It
follows that the overall number of calls is bounded by $O((L!)^2)$. 

Next, note that if some level-$i$ call $g(\mathcal{G},s,Q)$ is such that $|Q| < L$ and $|D(Q)| < n$, then line 5 computes the set $C'$, which can be achieved in $O(Ln)$ time by, for each $t \in [L]$, scanning over the components $C \in G_t$ (which collectively contain $n$ vertices) and adding a component $C \in G_t$ to $C'$ if and only if $|C-D(Q)| \geq (n-|D(Q)|)/(L-i)$. (Note that we can maintain a map from $V$ to $\{0,1\}$ that records for each vertex $v$ whether $v\in D(Q)$,
and hence the value $|C-D(Q)|$ can be computed in $O(|C|)$ time.) To compute the set $C^*$ in line 6 we can follow a similar approach: for each $t \in [L]-T(Q)$ ($|[L]-T(Q)| = L-i$), add a component $C \in G_t$ to $C^*$ if and only if it satisfies $C \in C'$. This requires $O((L-i)n) = O(Ln)$ time, and thus lines 5--6 take $O(Ln)$ time in total. Additionally, the return value of each recursive call is checked by the foreach loop (line 9) of its parent call in $O(1)$ time -- this contributes an extra $O((L!)^2)$ time over all recursive calls. On the other hand, if a call $g(\mathcal{G},s,Q)$ is such that $|Q| = L$ or $|D(Q)| = n$, then line 2 computes $W^?_\mathcal{G}(s,Q)$ in $O(Ln)$ time using Lemma~\ref{lem:wexpcorr}. Thus in all cases the overall work per recursive call is $O(Ln)$, and the total amount of time spent before $g(\mathcal{G},s,\emptyset)$ is returned is $O((L!)^2)\cdot O(Ln) = O(L(L!)^2n)$, as claimed.
\end{proof}

We remark that the algorithm of Theorem~\ref{thm:nstexpl} can be adapted to
\textsc{$k$-fixed NS-TEXP} in a straightforward way: If we are only interested
in visiting the vertices in a given set $X$ with $|X|=k$, an observation analogous to
Observation~\ref{obs:unvisit} shows the existence of a component $C$ that contains
at least a $1/(L-|T(Q)|$ fraction of the unvisited vertices in~$X$, i.e.,
$|(C-D(Q))\cap X|\ge (k-|D(Q)\cap X|)/(L-|T(Q)|)$. In Algorithm~\ref{alg:computeg},
we only need to replace the condition $|D(Q)|=n$ in lines 1 and~2 by
$|D(Q)\cap X|=k$, and the selection criterion for components in line 5
by $|(C-D(Q))\cap X|\ge (k-|D(Q)\cap X|)/(L-|T(Q)|)$.

\begin{corollary}
\textsc{$k$-fixed NS-TEXP} with parameter $L$ is in \textup{\textsf{FPT}}.
\end{corollary}

\subsection{\texorpdfstring{$\textup{\textsf{W}}[2]$}{W[2]}-hardness of Set NS-TEXP for parameter \texorpdfstring{$L$}{L}}
\label{ss:NSsetWTwo}%
Our aim in this section is to show that the \textsc{Set NS-TEXP} problem is $\textsf{W}[2]$-hard 
when parameterized by the lifetime $L$ of the input graph. The reduction is from the well-known 
\textsc{Set Cover} problem with parameter $k$ -- the maximum number of sets allowed in a 
candidate solution. \textsc{Set Cover} is 
known to be $\textsf{W}[2]$-hard for this parameterization~\cite{BPS_16}.
\begin{definition}[\textsc{Set Cover}]
	An instance of \textsc{Set Cover} is given as a tuple $(U, \mathcal{S}, k)$, 
	where $U=\{a_1,\dots,a_n\}$ is the ground set and $\mathcal{S} = \{S_1,\dots,S_m\}$ is 
	a set of subsets $S_i \subseteq U$. The problem then asks whether or not there exists a 
	subset $\mathcal{S'} \subseteq \mathcal{S}$ of size at most $k$ such that, for 
	all $i \in [n]$, there exists an $S \in \mathcal{S}'$ such that $a_i \in S$.
\end{definition}
For any instance $I$ of \textsc{Set Cover} that we consider, we will w.l.o.g. assume that for each $i \in [n]$ we have $a_i \in S_j$ for some $j \in [m]$.
\begin{theorem}
	\label{thm:nssethard}%
	\textsc{Set NS-TEXP} parameterized by $L$ (the lifetime of the 
	input non-strict temporal graph) is $\textup{\textsf{W}}[2]$-hard.
\end{theorem}
\begin{figure}
\centering
\resizebox{0.85\textwidth}{!}{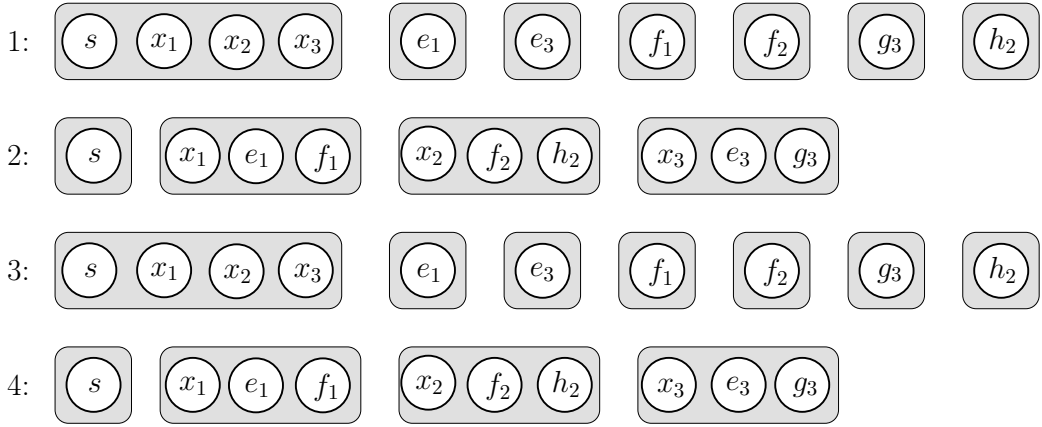}

\caption{%
Instance of \textsc{Set NS-TEXP} constructed from the
instance of \textsc{Set Cover} with $k=2$ given by
$U=\{e,f,g,h\}$ and $\mathcal{S}=\{S_1,S_2,S_3\}$
with $S_1= \{e,f\}$, $S_2=\{f,h\}$, $S_3=\{e,g\}$.
The set $\mathcal{X}$ of vertex subsets
that must be visited is
$\{ \{e_1,e_3\}, \{f_1,f_2\}, \{g_3\}, \{h_2\} \}$.}
\label{fig:nssethard}%
\end{figure}%
\begin{proof}
	Let $I=(U = \{a_1,\dots,a_n\}, \mathcal{S} = \{S_1,\dots,S_m\}, k)$ be an arbitrary 
	instance of \textsc{Set Cover} parameterized by $k$. We construct a corresponding 
	instance $I' = (\mathcal{G},s,\mathcal{X})$ of $\textsc{Set NS-TEXP}$ as follows:
	Let $V(\mathcal{G}) = \{s\} \cup \{x_j : j \in [m]\} \cup \{y_{i,j} : j \in [m], a_i \in S_j\}$, and
	define $X_i = \{y_{i,j} \in V(\mathcal{G}) : j \in [m]\}$ ($i \in [n]$) and 
	$\mathcal{X} = \bigcup_{i \in [n]} \{X_i\}$. We set the lifetime
	$L$ of $\mathcal{G}$ to $L = 2k$
	and specify the components for each 
	timestep $t \in [2k]$ as follows: In all odd steps let one component be 
	$\{s\} \cup \{x_j: j \in [m]\}$ and let all other vertices belong to components of 
	size~$1$. In even steps, for each $j \in [m]$ let there be a component 
	$\{y_{i,j} \in V(\mathcal{G}): i \in [n]\} \cup \{x_j\}$ and let~$s$
	form a component of size~$1$.
An example of the construction is shown in Figure~\ref{fig:nssethard}.
(In the figure, for the sake of readability, the elements of $U$ are
denoted by $e,f,g,h$ instead of $a_1,a_2,a_3,a_4$ and
the elements of $X_2$ are denoted by $f_2,h_2$ instead
of $y_{2,2},y_{4,2}$, and similarly for $X_1$ and $X_3$.)
Since $|V(\mathcal{G})| \leq 1 + m + mn = O(mn)$, $|\bigcup_{i \in [n]} X_i| = O(mn)$ 
and $L = 2k$ we have that the size of instance $I'$ is $|I'| = O(kmn)$ and the parameter 
$L$ is bounded solely by a function of instance $I$'s parameter $k$, as required. To complete 
the proof, we argue that $I$ is a yes-instance if and only if $I'$ is a yes-instance:

$(\implies)$ Assume that $I$ is a yes-instance; then there exists a collection of sets 
$\mathcal{S}'  \subseteq \mathcal{S}$ of size $|\mathcal{S}'| = k' \leq k$ and, for 
all $i \in [n]$, there exists $S \in \mathcal{S}'$ with $a_i \in S$. Let $S_{j_1},S_{j_2},\dots,S_{j_{k'}}$ 
be an arbitrary ordering of the sets in $\mathcal{S}'$; note that $j_i \leq m$ for all 
$i \in [k']$. We construct a non-strict temporal walk $W$ in $\mathcal{G}$ as follows: 
Starting at vertex $s$, for every $l \in [1,k']$, during timestep $t = 2l-1$ visit all 
vertices in the current component then finish timestep $2l-1$ positioned at $x_{j_l}$. 
The component occupied during step $2l$ will be the one containing $x_{{j}_l}$ -- explore 
all vertices contained in that component and finish step $2l$ positioned at $x_{{j}_l}$. 
If $k' < k$, then spend the steps of the interval $[2k'+1,2k]$ positioned in an arbitrary 
component. We claim that $W$ visits at least one vertex in $X_i$ for all $i \in [n]$. To see 
this, first note that for every $i \in [n]$ there exists an $S_j \in \mathcal{S}'$ such that 
$a_i \in S_j$.  Hence, by our reduction, it follows that a vertex $y_{i,j}$ is contained 
in the component containing $x_j$ during timestep $2l$ for every $l \in [k]$ and, by 
its construction, $W$ visits the component containing $x_j$ (and thus visits $y_{i,j} \in X_i$) 
during timestep $2l^*$ for some $l^*$ such that $j_{l^*} = j$. Since this holds for 
all $i \in [n]$ it follows that $W$ is a feasible solution and $I'$ is a yes-instance.

$(\impliedby)$ Assume that $I'$ is a yes-instance and that we have some non-strict 
temporal walk $W$ that visits at least one vertex in $X_i$ for all $i  \in [n]$. We 
first claim that $W$ visits any vertex of the form $y_{i,j}$ for the first time during
 an even step. To see this, observe that every $y_{i,j}$ lies disconnected in its own component 
 in every odd step $t$, and so to visit any $y_{i,j}$ in an odd step $W$ would need to occupy 
 the component containing $y_{i,j}$ during step $t-1$ and finish step $t-1$ positioned at 
 $y_{i,j}$; hence $y_{i,j}$ was already visited in step $t-1$, which is even. Therefore, 
 in order for $W$ to visit any $y_{i,j}$ it must be positioned, during at least one even 
 step, at the component containing $x_j$. Now, to construct a collection 
 of subsets $\mathcal{S}' \subseteq \mathcal{S}$ with size $x\le k$, let 
 $\mathcal{S}' = \{S_j : W \text{ visits the component containing } x_j \text{ during some even timestep}\}$. To see that $\mathcal{S}'$ is a cover 
 of $U$ with size $x \leq k$, observe that $W$ visits at least one vertex $y_{i,j}$ for every 
 $i \in [n]$; thus, by the reduction, for every $i \in [n]$ the element $a_i$ is contained in 
 set $S_j$ for some $S_j \in \mathcal{S}'$.
 It follows 
 that the union of $\mathcal{S}'$'s elements covers $U$, and so $I$ is a yes-instance.
\end{proof}

\section{Conclusion}
In this paper we have initiated the study of temporal exploration problems from the viewpoint
of parameterized complexity. For both strict and non-strict temporal walks, we have
shown several variants of the exploration problem to be
in \textsf{FPT}.
For the variant where we are given a family of vertex subsets and need to visit only one
vertex from each subset, we have shown $\textsf{W}[2]$-hardness for both the strict and
the non-strict model for parameter~$L$.
For non-strict temporal exploration, we have shown that the problem can be solved
in polynomial time if~$\gamma$, the maximum number of connected components per step, is bounded by~$2$.
An interesting question for future work is to determine whether \textsc{NS-TEXP} with parameter
$\gamma$ is in \textsf{FPT} or at least in \textsf{XP} (i.e., admits a polynomial-time
algorithm for each fixed value of~$\gamma$).
Another interesting question is whether \textsc{$k$-arbitrary NS-TEXP} is in \textsf{FPT}
for parameter~$L$.
\bibliographystyle{plain}
\bibliography{fpttexp}

\end{document}

%% file: free.pdf_t
\begin{picture}(0,0)%
\includegraphics{free.pdf}%
\end{picture}%
\setlength{\unitlength}{3947sp}%
\begingroup\makeatletter\ifx\SetFigFont\undefined%
\gdef\SetFigFont#1#2#3#4#5{%
  \reset@font\fontsize{#1}{#2pt}%
  \fontfamily{#3}\fontseries{#4}\fontshape{#5}%
  \selectfont}%
\fi\endgroup%
\begin{picture}(4377,1524)(10186,-4573)
\put(10576,-3436){\makebox(0,0)[lb]{\smash{{\SetFigFont{14}{16.8}{\familydefault}{\mddefault}{\updefault}{\color[rgb]{0,0,0}$i$:}%
}}}}
\put(10201,-4336){\makebox(0,0)[lb]{\smash{{\SetFigFont{14}{16.8}{\familydefault}{\mddefault}{\updefault}{\color[rgb]{0,0,0}$i+1$:}%
}}}}
\end{picture}%

%% file: restricted.pdf_t
\begin{picture}(0,0)%
\includegraphics{restricted.pdf}%
\end{picture}%
\setlength{\unitlength}{3947sp}%
\begingroup\makeatletter\ifx\SetFigFont\undefined%
\gdef\SetFigFont#1#2#3#4#5{%
  \reset@font\fontsize{#1}{#2pt}%
  \fontfamily{#3}\fontseries{#4}\fontshape{#5}%
  \selectfont}%
\fi\endgroup%
\begin{picture}(4377,1794)(10186,-4573)
\put(10576,-3436){\makebox(0,0)[lb]{\smash{{\SetFigFont{14}{16.8}{\familydefault}{\mddefault}{\updefault}{\color[rgb]{0,0,0}$i$:}%
}}}}
\put(10201,-4336){\makebox(0,0)[lb]{\smash{{\SetFigFont{14}{16.8}{\familydefault}{\mddefault}{\updefault}{\color[rgb]{0,0,0}$i+1$:}%
}}}}
\put(11101,-2986){\makebox(0,0)[lb]{\smash{{\SetFigFont{14}{16.8}{\familydefault}{\mddefault}{\updefault}{\color[rgb]{0,0,0}shrinking}%
}}}}
\put(13351,-2986){\makebox(0,0)[lb]{\smash{{\SetFigFont{14}{16.8}{\familydefault}{\mddefault}{\updefault}{\color[rgb]{0,0,0}growing}%
}}}}
\end{picture}%

%% file: nssethard.pdf_t
\begin{picture}(0,0)%
\includegraphics{nssethard.pdf}%
\end{picture}%
\setlength{\unitlength}{3947sp}%
\begingroup\makeatletter\ifx\SetFigFont\undefined%
\gdef\SetFigFont#1#2#3#4#5{%
  \reset@font\fontsize{#1}{#2pt}%
  \fontfamily{#3}\fontseries{#4}\fontshape{#5}%
  \selectfont}%
\fi\endgroup%
\begin{picture}(8127,3324)(4786,-3223)
\put(4801,-2986){\makebox(0,0)[lb]{\smash{{\SetFigFont{14}{16.8}{\familydefault}{\mddefault}{\updefault}{\color[rgb]{0,0,0}4:}%
}}}}
\put(4801,-2086){\makebox(0,0)[lb]{\smash{{\SetFigFont{14}{16.8}{\familydefault}{\mddefault}{\updefault}{\color[rgb]{0,0,0}3:}%
}}}}
\put(4801,-1186){\makebox(0,0)[lb]{\smash{{\SetFigFont{14}{16.8}{\familydefault}{\mddefault}{\updefault}{\color[rgb]{0,0,0}2:}%
}}}}
\put(4801,-286){\makebox(0,0)[lb]{\smash{{\SetFigFont{14}{16.8}{\familydefault}{\mddefault}{\updefault}{\color[rgb]{0,0,0}1:}%
}}}}
\put(5401,-262){\makebox(0,0)[lb]{\smash{{\SetFigFont{14}{16.8}{\familydefault}{\mddefault}{\updefault}{\color[rgb]{0,0,0}$s$}%
}}}}
\put(5926,-253){\makebox(0,0)[lb]{\smash{{\SetFigFont{14}{16.8}{\familydefault}{\mddefault}{\updefault}{\color[rgb]{0,0,0}$x_1$}%
}}}}
\put(7051,-266){\makebox(0,0)[lb]{\smash{{\SetFigFont{14}{16.8}{\familydefault}{\mddefault}{\updefault}{\color[rgb]{0,0,0}$x_3$}%
}}}}
\put(6499,-270){\makebox(0,0)[lb]{\smash{{\SetFigFont{14}{16.8}{\familydefault}{\mddefault}{\updefault}{\color[rgb]{0,0,0}$x_2$}%
}}}}
\put(8014,-262){\makebox(0,0)[lb]{\smash{{\SetFigFont{14}{16.8}{\familydefault}{\mddefault}{\updefault}{\color[rgb]{0,0,0}$e_1$}%
}}}}
\put(8902,-256){\makebox(0,0)[lb]{\smash{{\SetFigFont{14}{16.8}{\familydefault}{\mddefault}{\updefault}{\color[rgb]{0,0,0}$e_3$}%
}}}}
\put(9828,-288){\makebox(0,0)[lb]{\smash{{\SetFigFont{14}{16.8}{\familydefault}{\mddefault}{\updefault}{\color[rgb]{0,0,0}$f_1$}%
}}}}
\put(10728,-288){\makebox(0,0)[lb]{\smash{{\SetFigFont{14}{16.8}{\familydefault}{\mddefault}{\updefault}{\color[rgb]{0,0,0}$f_2$}%
}}}}
\put(11628,-261){\makebox(0,0)[lb]{\smash{{\SetFigFont{14}{16.8}{\familydefault}{\mddefault}{\updefault}{\color[rgb]{0,0,0}$g_3$}%
}}}}
\put(12502,-270){\makebox(0,0)[lb]{\smash{{\SetFigFont{14}{16.8}{\familydefault}{\mddefault}{\updefault}{\color[rgb]{0,0,0}$h_2$}%
}}}}
\put(5434,-1162){\makebox(0,0)[lb]{\smash{{\SetFigFont{14}{16.8}{\familydefault}{\mddefault}{\updefault}{\color[rgb]{0,0,0}$s$}%
}}}}
\put(6151,-1153){\makebox(0,0)[lb]{\smash{{\SetFigFont{14}{16.8}{\familydefault}{\mddefault}{\updefault}{\color[rgb]{0,0,0}$x_1$}%
}}}}
\put(6664,-1166){\makebox(0,0)[lb]{\smash{{\SetFigFont{14}{16.8}{\familydefault}{\mddefault}{\updefault}{\color[rgb]{0,0,0}$e_1$}%
}}}}
\put(7203,-1188){\makebox(0,0)[lb]{\smash{{\SetFigFont{14}{16.8}{\familydefault}{\mddefault}{\updefault}{\color[rgb]{0,0,0}$f_1$}%
}}}}
\put(8004,-1145){\makebox(0,0)[lb]{\smash{{\SetFigFont{14}{16.8}{\familydefault}{\mddefault}{\updefault}{\color[rgb]{0,0,0}$x_2$}%
}}}}
\put(8551,-1184){\makebox(0,0)[lb]{\smash{{\SetFigFont{14}{16.8}{\familydefault}{\mddefault}{\updefault}{\color[rgb]{0,0,0}$f_2$}%
}}}}
\put(9076,-1166){\makebox(0,0)[lb]{\smash{{\SetFigFont{14}{16.8}{\familydefault}{\mddefault}{\updefault}{\color[rgb]{0,0,0}$h_2$}%
}}}}
\put(9901,-1166){\makebox(0,0)[lb]{\smash{{\SetFigFont{14}{16.8}{\familydefault}{\mddefault}{\updefault}{\color[rgb]{0,0,0}$x_3$}%
}}}}
\put(10441,-1152){\makebox(0,0)[lb]{\smash{{\SetFigFont{14}{16.8}{\familydefault}{\mddefault}{\updefault}{\color[rgb]{0,0,0}$e_3$}%
}}}}
\put(10966,-1157){\makebox(0,0)[lb]{\smash{{\SetFigFont{14}{16.8}{\familydefault}{\mddefault}{\updefault}{\color[rgb]{0,0,0}$g_3$}%
}}}}
\put(5401,-2062){\makebox(0,0)[lb]{\smash{{\SetFigFont{14}{16.8}{\familydefault}{\mddefault}{\updefault}{\color[rgb]{0,0,0}$s$}%
}}}}
\put(5926,-2053){\makebox(0,0)[lb]{\smash{{\SetFigFont{14}{16.8}{\familydefault}{\mddefault}{\updefault}{\color[rgb]{0,0,0}$x_1$}%
}}}}
\put(7051,-2066){\makebox(0,0)[lb]{\smash{{\SetFigFont{14}{16.8}{\familydefault}{\mddefault}{\updefault}{\color[rgb]{0,0,0}$x_3$}%
}}}}
\put(6499,-2070){\makebox(0,0)[lb]{\smash{{\SetFigFont{14}{16.8}{\familydefault}{\mddefault}{\updefault}{\color[rgb]{0,0,0}$x_2$}%
}}}}
\put(8014,-2062){\makebox(0,0)[lb]{\smash{{\SetFigFont{14}{16.8}{\familydefault}{\mddefault}{\updefault}{\color[rgb]{0,0,0}$e_1$}%
}}}}
\put(8902,-2056){\makebox(0,0)[lb]{\smash{{\SetFigFont{14}{16.8}{\familydefault}{\mddefault}{\updefault}{\color[rgb]{0,0,0}$e_3$}%
}}}}
\put(9828,-2088){\makebox(0,0)[lb]{\smash{{\SetFigFont{14}{16.8}{\familydefault}{\mddefault}{\updefault}{\color[rgb]{0,0,0}$f_1$}%
}}}}
\put(10728,-2088){\makebox(0,0)[lb]{\smash{{\SetFigFont{14}{16.8}{\familydefault}{\mddefault}{\updefault}{\color[rgb]{0,0,0}$f_2$}%
}}}}
\put(11628,-2061){\makebox(0,0)[lb]{\smash{{\SetFigFont{14}{16.8}{\familydefault}{\mddefault}{\updefault}{\color[rgb]{0,0,0}$g_3$}%
}}}}
\put(12502,-2070){\makebox(0,0)[lb]{\smash{{\SetFigFont{14}{16.8}{\familydefault}{\mddefault}{\updefault}{\color[rgb]{0,0,0}$h_2$}%
}}}}
\put(5434,-2962){\makebox(0,0)[lb]{\smash{{\SetFigFont{14}{16.8}{\familydefault}{\mddefault}{\updefault}{\color[rgb]{0,0,0}$s$}%
}}}}
\put(6151,-2953){\makebox(0,0)[lb]{\smash{{\SetFigFont{14}{16.8}{\familydefault}{\mddefault}{\updefault}{\color[rgb]{0,0,0}$x_1$}%
}}}}
\put(6664,-2966){\makebox(0,0)[lb]{\smash{{\SetFigFont{14}{16.8}{\familydefault}{\mddefault}{\updefault}{\color[rgb]{0,0,0}$e_1$}%
}}}}
\put(7203,-2988){\makebox(0,0)[lb]{\smash{{\SetFigFont{14}{16.8}{\familydefault}{\mddefault}{\updefault}{\color[rgb]{0,0,0}$f_1$}%
}}}}
\put(8004,-2945){\makebox(0,0)[lb]{\smash{{\SetFigFont{14}{16.8}{\familydefault}{\mddefault}{\updefault}{\color[rgb]{0,0,0}$x_2$}%
}}}}
\put(8551,-2984){\makebox(0,0)[lb]{\smash{{\SetFigFont{14}{16.8}{\familydefault}{\mddefault}{\updefault}{\color[rgb]{0,0,0}$f_2$}%
}}}}
\put(9076,-2966){\makebox(0,0)[lb]{\smash{{\SetFigFont{14}{16.8}{\familydefault}{\mddefault}{\updefault}{\color[rgb]{0,0,0}$h_2$}%
}}}}
\put(9901,-2966){\makebox(0,0)[lb]{\smash{{\SetFigFont{14}{16.8}{\familydefault}{\mddefault}{\updefault}{\color[rgb]{0,0,0}$x_3$}%
}}}}
\put(10441,-2952){\makebox(0,0)[lb]{\smash{{\SetFigFont{14}{16.8}{\familydefault}{\mddefault}{\updefault}{\color[rgb]{0,0,0}$e_3$}%
}}}}
\put(10966,-2957){\makebox(0,0)[lb]{\smash{{\SetFigFont{14}{16.8}{\familydefault}{\mddefault}{\updefault}{\color[rgb]{0,0,0}$g_3$}%
}}}}
\end{picture}%